\documentclass[aps,pra,showpacs,superscriptaddress,groupedaddress,nofootinbib]{revtex4-2}

\pdfoutput=1

\input{embeddings.sty}

\begin{document}

\title{Catalytic Embeddings of Quantum Circuits}

\author{Matthew Amy}
\affiliation{School of Computing Science, Simon Fraser University, Burnaby, BC V5A 1S6, Canada}

\author{Matthew Crawford}
\affiliation{Department of Mathematics, University of North Carolina at Chapel Hill, Chapel Hill, NC 27599, USA}

\author{Andrew N. Glaudell}
\affiliation{Photonic Inc., Coquitlam, BC V3K 6T1, Canada}
\affiliation{Department of Mathematical Sciences, George Mason University, Fairfax, VA 22030, USA}
\affiliation{Quantum Science and Engineering Center, George Mason University, Fairfax, VA 22030, USA}
\affiliation{Booz Allen Hamilton, Washington, DC 20005, USA}

\author{Melissa L. Macasieb}
\affiliation{Booz Allen Hamilton, Washington, DC 20005, USA}

\author{Samuel S. Mendelson}
\affiliation{Strategic \& Computing Systems Department, Naval Surface Warfare Center, Dahlgren Division, Dahlgren, VA 22448, USA}

\author{Neil J. Ross}
\affiliation{Department of Mathematics and Statistics, Dalhousie University, Halifax, NS B3H 4R2, Canada}

\begin{abstract}
  If a set $\gatesetG$ of quantum gates is countable, then the
  operators that can be exactly represented by a circuit over
  $\gatesetG$ form a strict subset of the collection of all unitary
  operators. When $\gatesetG$ is universal, one circumvents this
  limitation by resorting to repeated gate approximations: every
  occurrence of a gate which cannot be exactly represented over
  $\gatesetG$ is replaced by an approximating circuit.  Here, we
  introduce catalytic embeddings, which provide an alternative to
  repeated gate approximations. With catalytic embeddings,
  approximations are relegated to the preparation of a fixed number of
  reusable resource states called catalysts.  Because the catalysts
  only need to be prepared once, when catalytic embeddings exist they
  always produce shorter circuits, in the limit of increasing gate
  count and target precision.  In the present paper, we lay the
  foundations of the theory of catalytic embeddings and we establish
  several of their structural properties. In addition, we provide
  methods to design catalytic embeddings, showing that their
  construction can be reduced to that of a single fixed matrix when
  the gates involved have entries in well-behaved rings of algebraic
  numbers. Finally, we showcase some concrete examples and
  applications. Notably, we show that catalytic embeddings generalize
  a technique previously used to implement the Quantum Fourier
  Transform over the Clifford+$T$ gate set with $O(n)$ gate
  approximations.
\end{abstract}

\maketitle

\section{Introduction}
\label{sec:intro}

Certain gate sets in quantum computing have become particularly
distinguished due to their rich mathematical structure, prominence in
proposals for scalable quantum computing, and practical
convenience. To perform tasks reliably on a quantum computer, it is
often important to understand the \emph{expressivity} of a given gate
set $\gatesetG$. This is done by characterizing the unitary operators
which can be exactly implemented as a circuit over $\gatesetG$.  In
the context of fault-tolerant quantum computation, such
characterizations have been possible through constructive
number-theoretic methods
\cite{kliuchnikov2013fast,Giles2013a,forest2015exact,Amy2020}.  These
characterizations of expressivity enabled major progress in the theory
of quantum circuits, including exact synthesis algorithms for circuits
on one or more qubits
\cite{kliuchnikov2013fast,giles2013remarks,Giles2013a,forest2015exact,Amy2020,glaudell2021optimal},
powerful circuit optimization strategies \cite{mosca2021polynomial},
and optimally efficient approximation methods
\cite{ross2016optimal,kliuchnikov2015framework}. Evidenced by this
growing body of work, number-theoretic characterizations of
expressivity have come to play a central role in a large number of
practical frameworks and strategies for quantum compilation.

Number-theoretic characterizations are particularly crucial in the
development of efficient approximation methods. Gate sets that are
well-suited for fault-tolerant quantum computation are finite and, as
a result, can only express a countable set of distinct operators. This
is in contrast to arbitrary unitary evolution, where operators may
come from an uncountable continuum. In order to perform universal
quantum computation, one therefore needs to rely on approximations to
extend the expressivity of a gate set: if an operator $U$ cannot be
implemented exactly over the gate set $\gatesetG$, then one looks for
a circuit $V$ over $\gatesetG$ which approximates $U$ for a desired
norm and precision. A number-theoretic characterization of the gate
set simplifies this process by allowing the operator $U$ to be
approximated by a matrix with entries in the number ring $\ringR$
associated with $\gatesetG$. In practice, this approximation process
is then repeated as often as the operator $U$ is needed.

Here, we revisit the problem of extending the expressivity of a gate
set. We introduce the method of \emph{catalytic embeddings} which
provides an alternative to repeated gate approximations in a number of
important contexts. For a unitary operation $U$ and a quantum gate set
$\gatesetG$, a catalytic embedding of $U$ over $\gatesetG$ is given by
a circuit $\phi_U$ over $\gatesetG$ and a quantum state $\ket{\chi_U}$
called a \emph{catalyst} such that, for all quantum states
$\ket{\psi}$ of appropriate dimension, we have
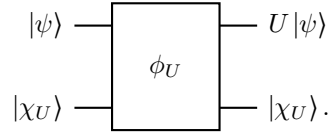
\begin{center}
    \begin{quantikz}
        \lstick{$\ket{\psi}$} & \gate[2]{\quad\phi_U\quad} & \qw\rstick{$U\ket{\psi}$}\\
        \lstick{$\ket{\chi_U}$} &  & \qw\rstick{$\ket{\chi_U}.$}
    \end{quantikz}
\end{center}

The power of such a construction becomes evident in the case where $U$
does not normally have an exact circuit representation over
$\gatesetG$.  Suppose we wish to implement the composite unitary
operation $V = G_k U G_{k-1} U \cdots G_1 U G_0$ to a target precision
$\varepsilon$, where each $G_j$ is a unitary with an exact
implementation over the gate set $\mathbb{G}$. If $V$ acts on more
than a few qubits and there are no additional obvious circuit
optimizations, we are left with few options besides replacing each $U$
independently.

The standard approach to this problem is to use repeated
approximations.  The operator $U$ is approximated by some circuit $U'$
over $\gatesetG$ up to $\varepsilon/k$, which then guarantees that the
circuit below implements $V$ to the desired precision.
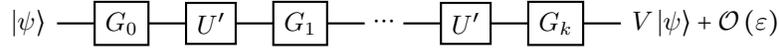
\begin{center}
\begin{quantikz}
    \lstick{$\ket{\psi}$} & \gate{G_0} & \gate{U'} & \gate{G_1} & \push{~\cdots~} & \gate{U'} & \gate{G_k} & \qw\rstick{$V\ket{\psi}+\order{\varepsilon}$}
\end{quantikz}
\end{center}
Using the Solovay-Kitaev algorithm, or an improved method if it is
available for $\gatesetG$, we can find such a $U'$ with
$\order{\log^c(k/\varepsilon)}$ gates in $\gatesetG$, where $c$ is a
constant that depends on $\gatesetG$ and the chosen approximation
technique. Writing $T$ for the gate count of the $k+1$ circuits $G_j$
after exact synthesis, the total gate count of the entire circuit is
given by
\[
    T + k\cdot\order{\log^c(k/\varepsilon)} = T + \order{k\log^c(k/\varepsilon)}.
\]

Now suppose that $\phi_U$ and $\ket{\chi_U}$ define a catalytic
embedding for $U$ over $\gatesetG$. Then, for any $\ket{\psi}$, we
have:
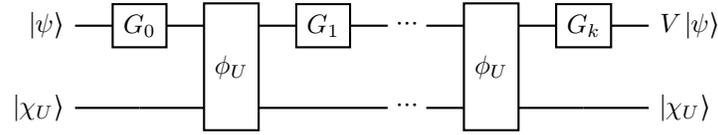
\begin{center}
\begin{quantikz}
    \lstick{$\ket{\psi}$} & \gate{G_0} & \gate[2]{\phi_U} & \gate{G_1} & \push{~\cdots~} & \gate[2]{\phi_U} & \gate{G_k} & \qw\rstick{$V\ket{\psi}$}\\
    \lstick{$\ket{\chi_U}$} & \qw & & \qw & \push{~\cdots~} & & \qw & \qw\rstick{$\ket{\chi_U}$}
\end{quantikz}
\end{center}
That is, in the presence of the catalyst $\ket{\chi_U}$, the above
circuit implements $V$ on its top register. Importantly, no elements
in the circuit are dependent on $\varepsilon$. In practice, we cannot
assume to have direct access to the catalyst $\ket{\chi_U}$, but we
can instead consider the action of this circuit on an approximate
catalyst $\ket{\chi}+\order{\delta}$. By linearity, we then have:
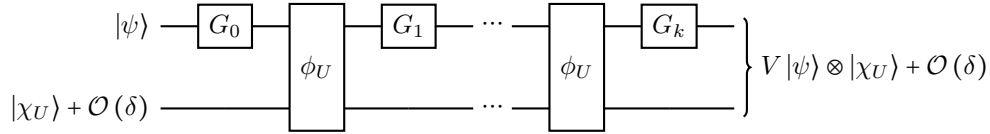
\begin{center}
\begin{quantikz}
    \lstick{$\ket{\psi}$} & \gate{G_0} & \gate[2]{\phi_U} & \gate{G_1} & \push{~\cdots~} & \gate[2]{\phi_U} & \gate{G_k} & \qw\rstick[wires=2]{$V\ket{\psi}\otimes\ket{\chi_U}+\order{\delta}$}\\
    \lstick{$\ket{\chi_U}+\order{\delta}$} & \qw & & \qw & \push{~\cdots~} & & \qw & \qw
\end{quantikz}
\end{center}
Thus, we get the desired action if we take $\delta=\varepsilon$, and
so we must approximate $\ket{\chi_U}$ to precision
$\varepsilon$. Writing $N$ for the gate count of the fixed circuit
$\phi_U$ and, as before, $T$ for the gate count of the $k+1$ circuits
$G_j$, and applying the Solovay-Kitaev algorithm (or an alternative) to
approximately prepare $\ket{\chi_U}$, we find the total gate count to
be
\[
    T + N k+\order{\log^c(1/\varepsilon)} = T + \order{k + \log^c(1/\varepsilon)}.
\]
Thus, given a circuit where $k$ approximations need to be made and an
overall precision of $\varepsilon$ is desired, catalytic embeddings
will incur the asymptotic cost savings
\[
    \order{k\log^c(k/\varepsilon)}\rightarrow \order{k + \log^c(1/\varepsilon)}.
\]
In other words, when catalytic embeddings exist, they will always
outperform repeated gate approximations, as $k$ and $1/\varepsilon$
increase.

In the present paper, we lay the foundations of the theory of
catalytic embeddings. We define a very general notion of catalytic
embedding and study its abstract properties. We then progressively
endow catalytic embeddings with additional structure, culminating in
the notion of \emph{linear catalytic embedding}. This provides an
axiomatization of a very general type of catalytic embedding that is
well-suited for gate sets admitting a number-theoretic
characterization, such as the Clifford+$T$, Clifford+$CS$, or
Toffoli-Hadamard gate sets. We then introduce a method to design
linear catalytic embeddings and we show that their construction can be
reduced to that of a single fixed matrix when the gates involved have
entries in well-behaved rings of algebraic numbers. Throughout, we
showcase concrete examples and applications.

Because catalytic embeddings have clear connections to algebraic
number theory, they provide a convenient framework through which to
study and compare gate sets that admit a number-theoretic
characterization. In the present paper, we define the framework of
catalytic embeddings and study its applications to the synthesis of
quantum circuits. We believe, however, that catalytic embeddings may
find applications in other areas of the theory of quantum circuits,
such as the study of the Clifford hierarchy \cite{cliffhierarchy} or
the classification of certain subgroups of the unitary group
\cite{aaronson2015classification,grier2022classification}.

The remainder of the paper is organized as follows. In
\cref{sec:prelims}, we define notation and provide the reader with a
very brief introduction to topics from algebraic number theory and the
theory of quantum computation. In \cref{sec:embedgates}, we formally
define catalytic embeddings. Afterwards, we introduce linear catalytic
embeddings in \cref{sec:lincats} and study their properties in
\cref{sec:props}. We introduce a particularly convenient method for
constructing linear catalytic embeddings of unitary matrices in
\cref{sec:standard}. Finally, we provide some concluding remarks and
outline avenues for future work in \cref{sec:conc}.

\section{Preliminaries}
\label{sec:prelims}

We begin by introducing some notation and terminology. We assume that
the reader has some familiarity with algebraic number theory and the
theory of quantum computation. Further details on these topics can be
found in \cite{artin2011algebra,schinzel_2000} and \cite{NC},
respectively.

\subsection{Kroneckerian Number Rings}
\label{ssec:kronrings}

Recall that a \emph{number field} is a finite degree field extension
of the field $\Q$ of rational numbers. An injective field homomorphism
from a number field $\fieldK$ into the field $\C$ of complex numbers
is an \emph{embedding} of $\fieldK$ into $\C$.

We will be interested in the number fields which can be endowed with
an unambiguous notion of complex conjugation. These are precisely the
\emph{Kroneckerian} number fields.

\begin{definition}[Kroneckerian Number Field]
  \label{def:kfield}
  A number field $\fieldK$ is \emph{Kroneckerian} if there exists an
  involution $c:\fieldK\to \fieldK$ such that, for any embedding
  $\sigma : \fieldK \to \C$, the following diagram commutes:
  \begin{center}
  \begin{tikzcd}[row sep=large,column sep=large,nodes={inner sep=5pt}]
    \fieldK \ar[r, "c"]\ar[d, swap, "\sigma"]& \fieldK  \ar[d, "\sigma"] \\
    \C\ar[r, swap, "(.)^\dagger"]& \C
  \end{tikzcd}  
  \end{center}
\end{definition}

The involution $c$ in \cref{def:kfield} is uniquely determined by the
number field $\fieldK$ and can be interpreted as the analogue of
complex conjugation in $\fieldK$. For this reason, and by a slight
abuse of notation, we use $(.)^\dagger$ to denote $c$ in what follows.

\begin{remark}
    \label{rem:kfield}
    Let $\fieldK$ be a number field. Then $\fieldK$ is said to be
    \emph{totally real} if $\sigma[\fieldK]\subseteq\R$ for every
    embedding $\sigma$ of $\fieldK$, and \emph{totally imaginary} if
    $\sigma[\fieldK]\not\subseteq\R$ for every embedding $\sigma$ of
    $\fieldK$. Moreover, $\fieldK$ is said to be \emph{CM} if it is a
    degree two totally imaginary extension of a totally real
    field. Kroneckerian number fields can also be defined as number
    fields that are either totally real or CM. The latter definition
    of Kroneckerian number fields is the standard one in the
    literature (see, e.g., \cite{schinzel_2000}) and is equivalent to
    \cref{def:kfield}. We choose to define Kroneckerian number fields
    as in \cref{def:kfield} because this definition emphasizes the
    fact that Kroneckerian number fields are those in which complex
    conjugation can be unambiguously defined.
\end{remark}

We will not only be interested in number fields but also in
\emph{number rings}, which we take to be subrings of number fields
(viewed as rings).  Recall that the field of fractions of a number
ring is the smallest field (up to isomorphism) containing that
ring. We introduce a notion of \emph{Kroneckerian number ring} to deal
with complex conjugation in number rings.

\begin{definition}[Kroneckerian Number Ring]
  \label{def:kring}
  A number ring $\ringR$ is \emph{Kroneckerian} if its field of
  fractions $\fieldfrac{\ringR}$ is a Kroneckerian field and if
  $r^\dagger \in \ringR$ for every $r\in\ringR$.
\end{definition}

An injective ring homomorphism from a number ring $\ringR$ into the
ring $\C$ of complex numbers is an \emph{embedding} of $\ringR$ into
$\C$.

We assume in the remainder of this paper that all number rings and
number fields exist in some ambient field $\mathcal{L}$, which is not
required to be Kroneckerian. For concreteness, one can consider all of
the rings and fields used in this paper as existing in the field of
algebraic numbers.

\subsection{Matrices and Unitary Groups}
\label{subsec:unitarygroups}

The addition, multiplication, direct sum, and tensor (or Kronecker)
product of matrices are defined as usual. If $\ringR$ is a ring, we
write $\matrices[\ringR][n]$ for the associative $\ringR$-algebra of
$n\times n$ matrices with entries in $\ringR$ and $\matrices[\ringR]$
for the collection of all square matrices with entries in
$\ringR$. That is,
\[
\matrices[\ringR] = \bigcup_{n>0} \matrices[\ringR][n].
\]

If $\ringR$ is a Kroneckerian number ring, we can extend conjugation
componentwise from $\ringR$ to $\matrices[\ringR]$. We then write
$U^\dagger$ for the conjugate transpose (or adjoint) of
$U\in\matrices[\ringR]$ and we say that $U$ is \emph{Hermitian} if
$U=U^\dagger$, \emph{unitary} if $U^{-1}=U^\dagger$, and \emph{normal}
if $UU^\dagger = U^\dagger U$.

For a fixed positive integer $n$, the collection of $n\times n$
unitary matrices over a Kroneckerian number ring $\ringR$ forms a
group known as the \emph{unitary group of degree $n$ over $\ringR$}.
        
\begin{definition}[Unitary Group]
  \label{def:unitarygroup}
  Let $\ringR$ be a Kroneckerian number ring. The
  \emph{unitary group of degree $n$ over $\ringR$}, denoted
  $\unitary[\ringR][n]$, is the group whose elements are
  $n$-dimensional unitary matrices over $\ringR$. That is, 
  \[
  \unitary[\ringR][n] = \s{ U \in \matrices[\ringR][n] \mid U^\dagger
  U = UU^\dagger = I}. 
  \]
  We write $\unitary[\ringR]$ for the collection of all 
  unitary matrices over $\ringR$. That is, 
  \[
  \unitary[\ringR] = \bigcup_{n>0} \unitary[\ringR][n] = \s{ U \in \matrices[\ringR] \mid U^\dagger
    U = UU^\dagger = I}.
  \]
\end{definition}

\cref{def:unitarygroup} is the reason for our interest in Kroneckerian
number rings. Indeed, because conjugation is well-defined in
Kroneckerian number rings, one can meaningfully talk about unitary
matrices and unitary groups over these rings.

Note that any embedding $\rho:\ringR \to \C$ extends componentwise to
a group homomorphism $\rho:\unitary[\ringR][n] \to \unitary[\C][n]$
satisfying $\rho(\unitary[\ringR][n]) =
\unitary[\rho(\ringR)][n]$. Note moreover that $\unitary[\ringR]$ is
closed under direct sum and tensor product. In other words, if
$U,V\in\unitary[\ringR]$, then $U\oplus V,U\otimes V\in
\unitary[\ringR]$.

\begin{remark}
\label{rem:ct1}
  The collection $\unitary[\ringR]$ is equipped with a partially
  defined operation of composition (the multiplication of matrices), a
  tensor product, and a \emph{dagger} (the conjugate transpose). These
  operations can be shown to endow $\unitary[\ringR]$ with the
  structure of a \emph{strict dagger symmetric monoidal groupoid}
  \cite{Mac,Selinger2011,HV2019}.
\end{remark}  

\subsection{Quantum States and Quantum Evolution}
\label{subsec:quantumstates}

\emph{Quantum states} are described by vectors in a Hilbert
space. Throughout the paper, we denote the Hilbert space of dimension
$m$ by $\hilbert_m$, assuming that the underlying field is $\C$ unless
otherwise stated.

We assume the existence of a preferred orthonormal basis called the
\emph{computational basis} and we represent the computational basis
for $\hilbert_m$ by the collection of quantum states $\basis_m =
\s{\ket{a}}[a\in\s{0,1,\dots,m-1}]$. We further assume that the
computational basis for the Hilbert space $\hilbert_m\otimes
\hilbert_n\cong\hilbert_{mn}$ is given by
\[
\s{\ket{a}\otimes\ket{b}=\ket{an+b}}[\ket{a}\in\basis_m,\ket{b}\in\basis_n].
\]

The evolution of a quantum state is described by a norm-preserving
linear transformation on a Hilbert space. Thus, in any orthonormal
basis, and in particular in the computational basis, this evolution
can be represented by a unitary matrix. A simple but ubiquitous
example of this kind of evolution is the transformation
$\swapalt(m,n):\hilbert_m\otimes\hilbert_n \to
\hilbert_n\otimes\hilbert_m$ whose action on computational basis
elements $\ket{a}\in\basis_m$ and $\ket{b}\in\basis_n$ is given by:
\[
\swapalt(m,n):\ket{a}\otimes\ket{b} \mapsto
\ket{b}\otimes\ket{a}.
\]

\subsection{Circuits}
\label{subsec:circs}

We now introduce \emph{circuits}, which are a widespread notation for
building complex linear operators in quantum information (see, e.g.,
\cite{NC}). Our presentation differs slightly from the standard one
and, in particular, emphasizes the distinction between a circuit and
the matrix it represents.

\begin{definition}[Gates]
  \label{def:gate}
  A \emph{gate} is a symbol. Every gate $G$ is associated with a
  positive integer $n$, called the \emph{dimension} of $G$, and an
  $n$-dimensional unitary matrix $\eval(G)$, called the
  \emph{evaluation} of $G$. A \emph{gate set} is a set of gates.
\end{definition}

\begin{definition}[Circuits]
  \label{def:circs}
  Let $\gatesetG$ be a gate set. \emph{Circuits over $\gatesetG$} and
  their \emph{dimension} are defined inductively as follows.
  \begin{itemize}
  \item For every $n\in\N^{*}$, $I_n$ is a circuit over $\gatesetG$ of
    dimension $n$.
  \item For every $m,n\in\N^{*}$, $\swapalt(m,n)$ is a circuit over
    $\gatesetG$ of dimension $mn$.
  \item If $G$ is a gate in $\gatesetG$ of dimension $n$, then $G$ is a circuit
    over $\gatesetG$ of dimension $n$.
  \item If $C$ and $D$ are circuits over $\gatesetG$ of dimension $n$,
    then $(C\circ D)$ is a circuit over $\gatesetG$ of dimension $n$.
  \item If $C$ and $D$ are circuits over $\gatesetG$ of dimensions $m$
    and $n$, respectively, then $(C\otimes D)$ is a circuit over
    $\gatesetG$ of dimension $mn$.
  \end{itemize}
  We write $\circuits(\gatesetG)$ for the collection of circuits over
  $\gatesetG$ and $\circuits_n(\gatesetG)$ for the circuits over
  $\gatesetG$ whose dimension is $n$.
\end{definition}

In \cref{def:circs}, $I_n$, $\swapalt(m,n)$, $\circ$, and $\otimes$
are treated as uninterpreted symbols. From this perspective, a circuit
over a gate set $\gatesetG$ is nothing but a word over the alphabet
\[
\gatesetG\cup \s{I_n \mid n\in\N^*} \cup \s{\swapalt(m,n) \mid
  m,n\in\N^*} \cup \s{\circ, \otimes, (,)}
\]
satisfying the constraints stated in \cref{def:circs}. For example, if
$\gatesetG=\s{A,B,C}$ and the gates $A$, $B$, and $C$, have dimension
1, 2, and 4, respectively, then $((I_2 \otimes((A\otimes I_2)\circ
B))\circ C)$ is a circuit over $\gatesetG$ of dimension 4. In what
follows, we always omit the outermost bracket around a circuit. That
is, we write $(I_2 \otimes((A\otimes I_2)\circ B))\circ C$ instead of
$((I_2 \otimes((A\otimes I_2)\circ B))\circ C)$.

By \cref{def:gate}, every gate in a gate set comes with an
evaluation. We now extend this evaluation from gates to circuits
containing these gates.

\begin{definition}[Evaluation]
  \label{def:evaluation}
  Let $\gatesetG$ be a gate set. The \emph{evaluation} function
  $e:\circuits(\gatesetG) \to \unitary(\C)$ is defined inductively as
  follows.
  \begin{itemize}
  \item For every $n\in\N$, $e(I_n)=I_n$.
  \item For every $m,n\in\N$, $e(\swapalt(m,n))=\swapalt(m,n)$.
  \item If $G$ is a gate in $\gatesetG$, then $e(G)=\eval(G)$.
  \item If $C$ and $D$ are circuits over $\gatesetG$, then $e(C\circ
    D)=e(C)\circ e(D)$.
  \item If $C$ and $D$ are circuits over $\gatesetG$, then
    $e(C\otimes D) = e(C) \otimes e(D)$.
  \end{itemize}
\end{definition}

\cref{def:gate,def:circs,def:evaluation} are meant to carefully
distinguish between a circuit (which is a word over some alphabet) and
the unitary matrix it denotes (its evaluation). Note that while
\cref{def:circs} treated $I_n$, $\swapalt(m,n)$, $\circ$, and
$\otimes$ as uninterpreted symbols, \cref{def:evaluation} fixes the
evaluation of these symbols to be natural one: $I_n$ is evaluated as
the identity matrix of dimension $n$, etc.

A common alternative to \cref{def:circs} is to denote circuits using
diagrams, rather than words. In diagrammatic notation, the circuit
$I_n$ is represented by a wire labelled by $n$ on both sides, and the
circuit $\swapalt(m,n)$ is represented by two crossing wires labelled
$m$ and $n$:
\begin{center}
\begin{tikzcd}
  \lstick{$n$} & \qw & \qw & \qw \rstick{$n$}
\end{tikzcd}
~~~ and ~~~
\begin{quantikz}[row sep = 0.5cm]
  \lstick{$m$} & \qw & \arrow[dl,dash,in=0,out=180] & \rstick{$n$} \qw\\
  \lstick{$n$} & \qw & \arrow[ul,dash,in=0,out=180] & \rstick{$m.$} \qw    
\end{quantikz}
\end{center}
For every $G \in \gatesetG$ of dimension $n$, the circuit $G$ is
represented by a box:
\begin{center}
    \begin{tikzcd}
        \lstick{$n$} & \gate{G} & \qw \rstick{$n$.}
    \end{tikzcd}
\end{center}
Finally, if $C$ and $D$ are circuits respectively represented as
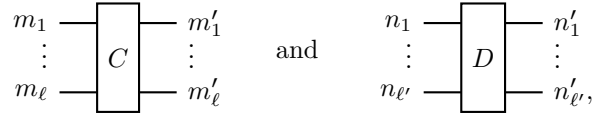
\begin{center}
   \begin{tikzcd}[row sep = 0.2cm]
     \lstick{$m_1$} & \gate[3,nwires={2}]{C} & \qw\rstick{$m_1'$}\\
     \lstick{$\rvdots$} & & \rstick{$\rvdots$} \\
     \lstick{$m_\ell$} &  & \qw\rstick{$m_\ell'$}        
   \end{tikzcd}
   ~~~ and ~~~
   \begin{tikzcd}[row sep = 0.2cm]
     \lstick{$n_1$} & \gate[3,nwires={2}]{D} & \qw\rstick{$n_1'$}\\
     \lstick{$\rvdots$} & & \rstick{$\rvdots$} \\     
     \lstick{$n_{\ell '}$} &  & \qw\rstick{$n_{\ell '}',$}
    \end{tikzcd}    
\end{center}
then the circuits $C\circ D$ and $C\otimes D$ are respectively
represented as
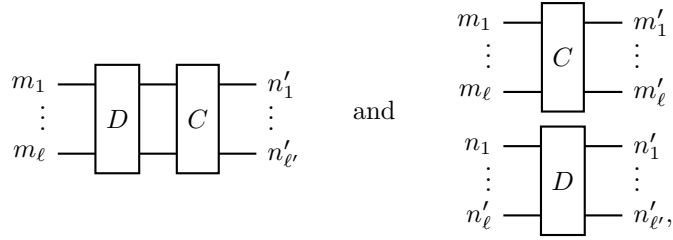
\begin{center}
   \begin{tikzcd}[row sep = 0.2cm]
        \lstick{$m_1$} & \gate[3,nwires={2}]{D} & \gate[3,nwires={2}]{C} & \qw\rstick{$n_1'$}\\
        \lstick{$\rvdots$} & & & \rstick{$\rvdots$} \\     
        \lstick{$m_\ell$} & & & \qw\rstick{$n_{\ell '}'$}
   \end{tikzcd}
   ~~~ and ~~~
   \begin{tikzcd}[row sep = 0.2cm]
        \lstick{$m_1$} & \gate[3,nwires={2}]{C} & \qw\rstick{$m_1'$}\\
        \lstick{$\rvdots$} & & \rstick{$\rvdots$} \\     
        \lstick{$m_\ell$} &  & \qw\rstick{$m_\ell'$} \\
        \lstick{$n_1$} & \gate[3,nwires={2}]{D} & \qw\rstick{$n_1'$}\\
        \lstick{$\rvdots$} & & \rstick{$\rvdots$} \\     
        \lstick{$n_\ell'$} &  & \qw\rstick{$n_{\ell '}',$}        
    \end{tikzcd}    
\end{center}
where $C\circ D$ is well-defined if and only if $\ell = \ell'$ and
$m_j'=n_j$ for every $1\leq j \leq \ell$. We omit the wire labels if
they are unimportant or can be inferred from context.

\begin{example}
  \label{ex:circs2}
  The $\s{X, CX, CCX}$ gate set consists of the gates $X$, $CX$, and $CCX$, where
  \[
  e(X) = \begin{bmatrix} 0 & 1 \\ 1 & 0 \end{bmatrix}
  \quad 
  e(CX) = \diag(I_2,e(X))
  \quad
  \mbox{and}
  \quad
  e(CCX) =\diag(I_4, e(CX)).
  \]
\end{example} 

\begin{example}
  \label{ex:circs1}
  The \emph{Clifford+$T$} gate set consists of the gates $H$, $T$, and
  $CX$, where $X \coloneqq H\circ (T\circ(T\circ(T\circ (T\circ H))))$
  and
  \[
  e(H)= \frac{1}{\sqrt 2}
  \begin{bmatrix}
    1 & 1 \\ 1 & -1
  \end{bmatrix}, \qquad
  e(T)=
  \begin{bmatrix}
    1 & 0 \\ 0 & e^{i\pi/4}
  \end{bmatrix}, \qquad \mbox{and} \qquad
  e(CX)= \diag(I_2,e(X))
  \]
  The Clifford+$T$ circuit $C=((I_2\otimes H) \circ CX) \circ
  (T\otimes T)$ can be diagrammatically represented as follows (using
  the standard convention for the $CX$ gate)
  
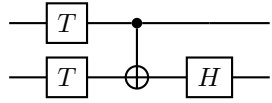
\begin{center}
    \begin{tikzcd}[row sep = 0.2cm]
        & \gate{T} & \ctrl{1} & \qw      & \qw \\
        & \gate{T} & \targ{}  & \gate{H} & \qw         
    \end{tikzcd}
  \end{center}
  and the circuit $C$ evaluates to 
  \[
  e(C)= \frac{1}{\sqrt 2}
  \begin{bmatrix}
    1 & e^{i\pi/4}  & 0         & 0 \\
    1 & -e^{i\pi/4} & 0         & 0 \\
    0 & 0         & e^{i\pi/4}  & e^{i\pi/2} \\
    0 & 0         & -e^{i\pi/4} & e^{i\pi/2} \\     
  \end{bmatrix}.
  \]
\end{example} 

\begin{remark}
  \label{rem:circdifs1}
  The notion of circuit introduced in this section differs from the
  usual one (as defined, say, in \cite{NC}). Firstly, the circuits
  introduced here can act (when evaluated) on Hilbert spaces of
  arbitrary dimensions, whereas circuits in the quantum computing
  literature are often restricted to spaces of dimension $2^n$ for
  some nonnegative integer $n$. The reader familiar with quantum
  circuits can think of the circuits described above as acting not
  only on qubits but, more generally, on a mixture of qudits of
  varying (finite) dimensions. Secondly, circuits are traditionally
  considered up to certain transformations. For example, according to
  \cref{def:circs,def:evaluation}, if $A$ and $B$ are two gates of
  dimension 2, then $A\otimes B$, $(A\otimes I_2)\circ(I_2 \otimes
  B)$, and $(I_2\otimes B) \circ (A \otimes I_2)$ are three distinct
  circuits representing the same matrix. In contrast, circuits in the
  literature are often defined in a way that equates these three
  circuits.
\end{remark}

\begin{definition}
  \label{def:unitarygateset}
  Let $\circuitsalt$ be a collection of circuits. We define
  $\unitary(\circuitsalt)$ as the image of $\circuitsalt$ under the
  evaluation function $e$. That is,
  $\unitary(\circuitsalt)=e[\circuitsalt]$. We further define
  $\unitary_n(\circuitsalt)$ as
  $\unitary[\circuitsalt][n]=\unitary[\circuitsalt]\cap\unitary[\C][n]$.
\end{definition}

The set $\unitary(\circuits(\gatesetG))$ is the collection of all
unitary matrices that can be represented by a circuit over
$\gatesetG$. If $\mathscr{U}$ is a collection of unitary matrices,
then we sometimes interpret $\mathscr{U}$ as a gate set by
introducing, for every $U\in\mathscr{U}$ of dimension $n$, a gate
$G_U$ of dimension $n$ with $e(G_U)=U$. In such a case, we then have
$\mathscr{U} \subseteq \unitary(\circuits(\mathscr{U}))$. If
$\mathscr{U}$ is closed under composition and tensor products, and
contains all identities and swaps, then we in fact have $\mathscr{U} =
\unitary(\circuits(\mathscr{U}))$.

By definition, the evaluation function $e:\circuitsalt \to
\unitary(\circuitsalt)$ is surjective. It is not injective, however,
since many different circuits evaluate to the same unitary matrix. An
\emph{exact synthesis function for $\circuitsalt$} is a function which
assigns, to each unitary matrix in $\unitary(\circuitsalt)$, a unique
circuit representing that matrix.

\begin{definition}[Exact Synthesis]
  \label{def:exactsynth}
  An \emph{exact synthesis function} for a collection of circuits
  $\circuitsalt$ is a function $s:\unitary(\circuitsalt)\to
  \circuitsalt$ such that $e\circ s = I_{\unitary(\circuitsalt)}$. An
  algorithm computing an exact synthesis function is an \emph{exact
  synthesis algorithm}.
\end{definition}

The notion of exact synthesis introduced in \cref{def:exactsynth} is
sometimes known as \emph{ancilla-free} exact synthesis.

\begin{example} 
\label{ex:exactsynth}
Several exact synthesis algorithms have been introduced in the
literature. In particular, exact synthesis algorithms exist for the
gate sets $\s{X, CX, CCX}$ and $\s{H, T, CX}$ discussed in
\cref{ex:circs1,ex:circs2} (see \cite{Giles2013a,shende03}). In both
cases, the exact synthesis algorithm relies on (and in fact
establishes) a characterization of the unitary matrices that can be
represented over the gate set. In particular, for $n\geq 4$,
\begin{itemize}
\item the elements of $\unitary_{2^n}(\circuits(\s{X, CX, CCX}))$ are
  exactly the permutations matrices of dimension $2^n$ that have
  determinant 1 (i.e., $\unitary_{2^n}(\circuits(\s{X, CX, CCX})) =
  A_{2^n}$), and
\item the elements of $\unitary_{2^n}(\circuits(\s{H, T, CX}))$ are
  exactly the elements of $\unitary_{2^n}(\Z[1/2,e^{2\pi i/8}])$ that
  have determinant 1.
\end{itemize}
The exact synthesis functions of \cite{Giles2013a,shende03}, and the
accompanying exact synthesis algorithms, can be adapted to $n<4$ by
varying the condition on the determinant. In fact, as we will discuss
below, the conditions on the determinant can be lifted altogether
through the use of ancillas.
\end{example}

\section{Catalytic Embeddings}
\label{sec:embedgates}

Recall from \cref{sec:intro} that we are interested in the problem of
constructing a circuit that, in the presence of a well-chosen resource
state, implements a desired operator. To make this intuition precise,
we now introduce \emph{catalytic embeddings}.

\begin{definition}
  \label{def:catembedcircs}
  Let $C$ be an $n$-dimensional circuit and let $\circuitsalt$ be a
  collection of circuits. A \emph{$k$-dimensional catalytic embedding
  of $C$ in $\circuitsalt$} is a pair $(\phi,\projector)$ consisting
  of an $(nk)$-dimensional circuit $\phi\in\circuitsalt$ and a nonzero
  orthogonal projector $\projector$ on $\hilbert_k$ satisfying the
  following \emph{catalytic condition}:
  \[
    e(\phi)(I\otimes\projector)=e(C)\otimes \projector.
  \]
\end{definition}

If $(\phi, \projector)$ is a $k$-dimensional catalytic embedding of a
circuit $C$, and $\ket{\chi}$ is such that
$\projector\ket{\chi}=\ket{\chi}$, then the catalytic condition
ensures that for any $\ket{\psi}$ we have
\begin{equation}
\label{eq:catstate}
e(\phi)\ket{\psi}\ket{\chi} = e(\phi)(I\otimes \projector)\ket{\psi}\ket{\chi}=(e(C)\otimes \projector) \ket{\psi}\ket{\chi} = (e(C)\ket{\psi})(\projector\ket{\chi})= (e(C)\ket{\psi})\ket{\chi}.
\end{equation}
The circuit $\phi$ therefore acts as $C$ in the presence of
$\ket{\chi}$. Moreover, $\ket{\chi}$ remains unchanged by the
application of $\phi$. For these reasons, and in accordance with
related work \cite{Beverland_2020}, we refer to any quantum state
$\ket{\chi}\in\hilbert_k$ with $\projector\ket{\chi}=\ket{\chi}$ as a
\emph{catalyst for $C$ over $\circuitsalt$}. Moreover, if $(\phi,
\projector)$ is an embedding of a circuit $C$, we sometimes refer to
$\phi$ as the \emph{embedding} and to $\projector$ as the
\emph{catalytic projector}. Note that a catalytic embedding can be
constructed from a catalyst. Indeed, if $\phi$ and $\ket{\chi}$
jointly satisfy \cref{eq:catstate} and $\ket{\chi}$ is a unit vector,
then $(\phi, |\chi\rangle\langle \chi|)$ is a catalytic embedding of
$C$.

As the propositions below show, the (evaluation of the) embedding of a
circuit $C$ can always be written as a block-diagonal matrix, up to a
change of basis.

\begin{proposition}
    \label{prop:leftside}
    Let $C$ be a circuit and $\circuitsalt$ be a collection of
    circuits. If $(\phi,\projector)$ is a catalytic embedding of $C$
    in $\circuitsalt$, then $(I\otimes\projector)e(\phi) =
    e(C)\otimes\projector$.
\end{proposition}

\begin{proof}
    Note that $(e(C) \otimes I)(I\otimes \projector) = e(C)\otimes
    \projector = (I\otimes \projector)(e(C) \otimes I)$. The catalytic
    condition therefore yields
    \[
    e(\phi)(I\otimes \projector)(e(C)^\dagger \otimes I) = I\otimes \projector.
    \]
    Since $\projector$ is orthogonal, applying $(.)^\dagger$ on both
    sides of the above equation gives $(e(C) \otimes I)(I\otimes
    \projector)e(\phi)^\dagger = I\otimes \projector$, which then
    implies the desired equation by left-multiplication with
    $e(\phi)$.
\end{proof}

\begin{proposition}
  \label{prop:directsumform}
  Let $C$ be a circuit and $\circuitsalt$ be a collection of
  circuits. If $(\phi,\projector)$ is a catalytic embedding of $C$ in
  $\circuitsalt$, then there exists a unitary $U$ such that
  \[
      e(\phi) \sim \underbrace{e(C)\oplus e(C)\oplus\cdots\oplus e(C)}_{\rank(\projector)}\oplus U,
  \]
  where $\sim$ denotes equality up to conjugation by a unitary. 
\end{proposition}

\begin{proof}
    Suppose that $C$ is an $n$-dimensional circuit and that
    $(\phi,\projector)$ is a $k$-dimensional catalytic embedding, so
    that $e(\phi)$ acts on $\hilbert = \hilbert_{nk}$. Consider the
    projectors $P=I\otimes \projector$ and $P^\perp = I-P$. Note that
    $PP^\perp = P^\perp P =0$. Note moreover that, by the catalytic
    condition and \cref{prop:leftside}, $Pe(\phi) = e(\phi)P$. We can
    therefore decompose the action of $e(\phi)$ on $\hilbert$ into its
    action on the subspaces associated with $P$ and $P^\perp$ as
    follows:
    \begin{align*}
        e(\phi) &= (P+P^{\perp})e(\phi)(P+P^{\perp}) \\
        & = Pe(\phi)P + Pe(\phi)P^\perp + P^\perp e(\phi)P + P^\perp e(\phi)P^\perp \\
        & = Pe(\phi)P + e(\phi)PP^\perp + P^\perp Pe(\phi) + P^\perp e(\phi)P^\perp \\        
        & = P e(\phi) P + P^{\perp} e(\phi) P^{\perp},
    \end{align*}
    Now let $r$ be the rank of $\projector$ and let
    $Q=\diag(I_r,0_{k-r})$. By the spectral theorem, there exists a
    unitary matrix $V$ such that $V\projector V^\dagger = Q$. Let
    $W=(V\otimes I_n)\circ \swapalt(n,k)$. Then $W$ is unitary and
    \[
    W e(\phi) W^\dagger = WP e(\phi) PW^\dagger + WP^{\perp} e(\phi) P^{\perp}W^\dagger
    = Q\otimes e(C) + [(I-Q)\otimes I] W e(\phi)W^\dagger [(I-Q)\otimes I],
    \]
    since $Pe(\phi)P= e(\phi)P = e(C)\otimes \projector$. Thus $W
    e(\phi) W^\dagger = \diag(I_r\otimes e(C), U)$ for some matrix
    $U$. Since $e(\phi)$, $I_r\otimes e(C)$, and $W$ are unitary, $U$
    must also be unitary.
\end{proof}

Rather than focus on the embedding of a single circuit, we will focus
on the embedding of collections of circuits.

\begin{definition}
  \label{def:catembedcircsets}
  Let $\circuitsalt$ and $\circuitsalt'$ be two collections of
  circuits. A \emph{catalytic embedding of $\circuitsalt$ in
  $\circuitsalt'$} is a collection
  $\s{(\phi_C,\projector_C)}[C\in\circuitsalt]$, where, for each
  $C\in\circuitsalt$, the pair $(\phi_C,\projector_C)$ is a catalytic
  embedding of $C$ in $\circuitsalt'$.
\end{definition}  

A catalytic embedding $\s{(\phi_C,\projector_C)}[C\in\circuitsalt]$ as
in \cref{def:catembedcircsets} implicitly specifies two functions. The
function $\phi:C\mapsto \phi_C$, which assigns an embedding to every
circuit in $\circuitsalt$, and the function $\projector:C\to
\projector_C$, which assigns a catalytic projector to every circuit in
$\circuitsalt$. For brevity, we sometimes write $(\phi, \projector):
\circuitsalt\to\circuitsalt'$ to refer to the catalytic embedding
$\s{(\phi_C,\projector_C)}[C\in\circuitsalt]$ of $\circuitsalt$ in
$\circuitsalt'$.

\begin{definition}
  \label{def:hom}
  We say that a catalytic embedding $(\phi, \projector):
  \circuitsalt\to \circuitsalt'$ is \emph{homogeneous} when
  $\projector_C=\projector_D$ for every $C,D\in\circuitsalt$.
\end{definition}

If $(\phi, \projector):\circuitsalt \to \circuitsalt'$ is a
homogeneous catalytic embedding, then there exists an integer $k$ such
that, for every $C\in\circuitsalt$, $(\phi_C, \projector_C)$ is a
$k$-dimensional catalytic embedding. We then call $k$ the
\emph{dimension} of $(\phi,\projector)$. A catalytic embedding
$(\phi,\projector):\circuitsalt \to \circuitsalt'$ can always be
homogenized when $\circuitsalt$ is finite. This can be done, for
example, by replacing $\projector$ by $\bigotimes \projector_C$ and
amending $\phi$ appropriately.

\begin{example}
\label{ex:cat-exactsynth}
An exact synthesis function in the sense of \cref{def:exactsynth} is a
1-dimensional catalytic embedding of $\unitary(\circuitsalt)$ (viewed
as a gate set) in $\circuitsalt$. More generally, exact synthesis
results that rely on ancillas, such as the most general algorithms
introduced in \cite{shende03,Giles2013a}, can also be viewed as
catalytic embeddings: \cite{shende03} introduces a 2-dimensional
catalytic embedding of $S_{2^n}$ in $\circuits_{2^{n+1}}(\s{X, CX,
  CCX})$, and \cite{Giles2013a} introduces a 2-dimensional catalytic
embedding of $\unitary_{2^n}(\Z[1/2,e^{2\pi i/8}])$ in
$\circuits_{2^{n+1}}(\s{H, T, CX})$. In both cases, the embeddings are
homogeneous. The embedding of \cite{Giles2013a} uses the projector
$|0\rangle\langle 0|$, whereas that of \cite{shende03} utilizes the
projector $I_2$.
\end{example}

\begin{definition}
  \label{def:concat}
  The \emph{concatenation} of two catalytic embeddings $(\phi,
  \projector): \circuitsalt\to \circuitsalt'$ and $(\phi',
  \projector'): \circuitsalt'\to \circuitsalt''$ is the catalytic
  embedding $(\phi', \projector')\circ (\phi,\projector) :
  \circuitsalt \to \circuitsalt''$ defined by
  \[
  (\phi', \projector')\circ (\phi,\projector) = (\phi'\circ \phi, \projector \otimes \projector').
  \]
\end{definition}

It is straightforward to verify that the concatenation of two
catalytic embeddings is indeed a catalytic embedding. If
$(\phi,\projector)$ and $(\phi',\projector')$ are two catalytic
embeddings of dimension $k$ and $k'$, respectively, then $(\phi',
\projector')\circ (\phi,\projector)$ is a catalytic embedding of
dimension $kk'$. Moreover, the concatenation of embeddings is
associative and preserves homogeneity.

We will be especially interested in catalytic embeddings when
$\circuitsalt = \circuits(\gatesetG)$ and $\circuitsalt' =
\circuits(\gatesetH)$, for some gate sets $\gatesetG$ and
$\gatesetH$. One often thinks of the gates in a gate set $\gatesetG$
as generators for the circuits in $\circuits(\gatesetG)$. It is
therefore natural to expect that a catalytic embedding of $\gatesetG$
in some collection of circuits $\circuitsalt'$ might extend to a
catalytic embedding of $\circuits(\gatesetG)$ in $\circuitsalt'$. This
is indeed the case, as the following example illustrates.

\begin{example}
\label{ex:extension}
Consider the following circuit $F$ over some gate set $\gatesetG =
\s{A,B,C,D}$:
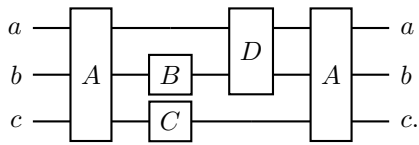
\begin{center}
\begin{tikzcd}[row sep=0.1cm]
\lstick{$a$} &\gate[3]{A} &\qw      &\gate[2]{D} &\gate[3]{A} &\qw\rstick{$a$}\\
\lstick{$b$} &            &\gate{B} &            &            &\qw\rstick{$b$}\\
\lstick{$c$} &            &\gate{C} &\qw         &            &\qw\rstick{$c$.}
\end{tikzcd}
\end{center}
Suppose that we have a catalytic embedding of $\gatesetG$ in
$\circuitsalt$ such that $(\phi_A,\projector_1)$,
$(\phi_B,\projector_1\otimes\projector_3)$,
$(\phi_C,\projector_1\otimes\projector_2)$, and
$(\phi_D,\projector_1\otimes\projector_3\otimes\projector_2)$ are
catalytic embeddings of $A,B,C$, and $D$ in $\circuitsalt$ with
$\projector_1:\hilbert_d\rightarrow\hilbert_d$,
$\projector_2:\hilbert_e\rightarrow\hilbert_e$, and
$\projector_3:\hilbert_f\rightarrow\hilbert_f$. We can then embed $F$
in $\circuitsalt$ as $(\phi,\projector)$ where $\projector =
\projector_1\otimes \projector_2 \otimes \projector_3$ and $\phi$ is
the circuit below.
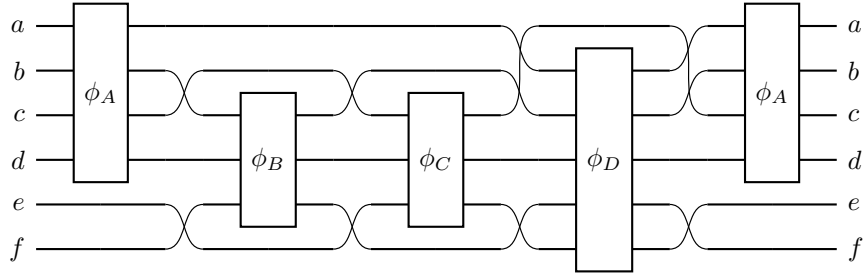
\begin{center}
\begin{tikzcd}[row sep=0cm]
\lstick{$a$} &\gate[4]{\phi_A} &\qw &\qw    &\qw              &\qw &\qw   &\qw              &\qw &\myddl &\qw              &\qw &\mydl  &\gate[4]{\phi_A} &\qw\rstick{$a$}\\
\lstick{$b$} &                 &\qw &\mydl  &\qw              &\qw &\mydl &\qw              &\qw &\myul  &\gate[5]{\phi_D} &\qw &\mydl  &                 &\qw\rstick{$b$}\\
\lstick{$c$} &                 &\qw &\myul  &\gate[3]{\phi_B} &\qw &\myul &\gate[3]{\phi_C} &\qw &\myul  &                 &\qw &\myuul &                 &\qw\rstick{$c$}\\
\lstick{$d$} &                 &\qw &\qw    &                 &\qw &\qw   &                 &\qw &\qw    &                 &\qw &\qw    &                 &\qw\rstick{$d$}\\
\lstick{$e$} &\qw              &\qw &\mydl  &                 &\qw &\mydl &                 &\qw &\mydl  &                 &\qw &\mydl  &\qw              &\qw\rstick{$e$}\\
\lstick{$f$} &\qw              &\qw &\myul  &\qw              &\qw &\myul &\qw              &\qw &\myul  &                 &\qw &\myul  &\qw              &\qw\rstick{$f$}
\end{tikzcd}
\end{center}
\end{example}

A catalytic embedding $(\phi,\projector):\gatesetG \to \circuitsalt$
can always be extended to a catalytic embedding $\circuits(\gatesetG)
\to \circuitsalt$, e.g., as in \cref{ex:extension}. Such an extension
takes a particularly nice form when the catalytic embedding
$(\phi,\projector)$ is homogeneous. In this case, the same catalyst
can be used for any circuit over $\gatesetG$.

\begin{definition}
  \label{def:inducedcatembedcirc}
  Let $(\phi,\projector):\gatesetG \to \circuitsalt$ be a homogeneous
  catalytic embedding of dimension $k$ of a gate set $\gatesetG$ in a
  collection of circuits $\circuitsalt$ and let $C$ be a circuit over
  $\gatesetG$. Then the \emph{catalytic embedding of $C$ in
  $\circuitsalt$ induced by $(\phi,\projector)$} is the pair
  $(\overline{\phi}_C, \overline{\projector}_C)$ where
  $\overline{\projector}_C = \projector$ and $\overline{\phi}$ is
  defined by induction on $C$ as follows.
  \begin{itemize}
  \item If $C=I_n$ for some $n\in\N$, then $\overline{\phi}_C=I_n\otimes I_k$.
  \item If $C=\swapalt(n,m)$ for some $n,m\in\N$, then $\overline{\phi}_C=\swapalt(n,m)\otimes I_k$.
  \item If $C=G$ for some $G\in\gatesetG$, then $\overline{\phi}_C=\phi_G$.
  \item If $C=(C_1\circ C_2)$, then $\overline{\phi}_C=(\overline{\phi}_{C_1} \circ \overline{\phi}_{C_2})$.
  \item If $C= (C_1\otimes C_2)$, with $C_1$ of dimension $m$ and $C_2$ of dimension $n$, then 
  \[
  \overline{\phi}_C = 
  (((I_m\otimes \phi_{C_2}) 
  \circ (I_m\otimes \swapalt(k,n)))
  \circ (\phi_{C_1} \otimes I_n))
  \circ (I_m \circ \swapalt(n,k)).
  \]
  \end{itemize}
  We write $\overline{(\phi,\projector)}$ for the catalytic embedding
  $\s{(\overline{\phi}_C,\overline{\projector}_C)\mid C\in
    \circuits(\gatesetG)}:\circuits(\gatesetG)\to\circuitsalt$.
\end{definition}

It can be verified that the induced catalytic embedding introduced in
\cref{def:inducedcatembedcirc} is a well-defined homogeneous catalytic
embedding.

\begin{example}
\label{ex:extension2}
Consider the circuit $F$ from \cref{ex:extension} and assume that $F$
was given as
\[
F = A \circ ((D\otimes I_c) \circ ((I_a \otimes (B \otimes C)) \circ A)).
\]
We can define a homogeneous catalytic embedding of $\gatesetG$ in
$\circuitsalt$ of dimension $g = d\cdot e\cdot f$ via
\[
    (\phi',\projector) = \left\{(\phi_A',\projector),(\phi_B',\projector),(\phi_C',\projector),(\phi_D',\projector)\right\}
\]
where $\projector =
\projector_1\otimes\projector_2\otimes\projector_3$ and
\begin{align*}
    \phi_A' &= \phi_A\otimes I_{ef}\\
    \phi_B' &= (I_{b d}\otimes\swapalt(f,e))\circ((\phi_B\otimes I_e)\circ(I_{b d}\otimes\swapalt(e,f)))\\
    \phi_C' &= \phi_C\otimes I_f\\
    \phi_D' &= (I_{a b d}\otimes\swapalt(f,e))\circ(\phi_D\circ(I_{a b d}\otimes\swapalt(e,f))).
\end{align*}
The corresponding circuit diagram for the image of $F$ under the
action of the catalytic embedding induced by $(\phi',\projector)$ is
the circuit depicted below.
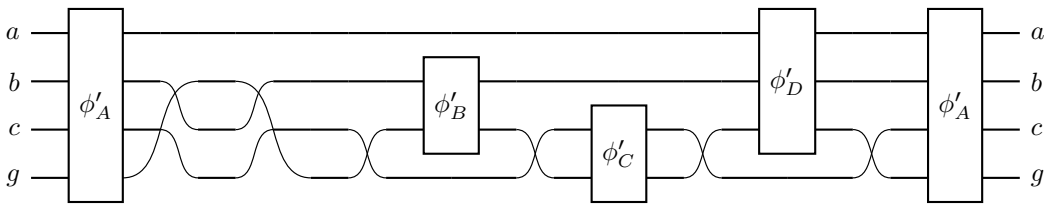
\begin{center}
\begin{tikzcd}[row sep=0cm]
\lstick{$a$}&\gate[4]{\phi_A'} &\qw &\qw     &\qw &\qw   &\qw     &\qw &\qw   &\qw              &\qw &\qw   &\qw              &\qw &\qw   &\gate[3]{\phi_D'} &\qw &\qw   &\gate[4]{\phi_A'} &\qwrs{$a$}\\
\lstick{$b$}&                 &\qw &\myddll &\qw &\mydl &\qw     &\qw &\qw   &\gate[2]{\phi_B'} &\qw &\qw   &\qw              &\qw &\qw   &                 &\qw &\qw   &                 &\qwrs{$b$}\\
\lstick{$c$}&                 &\qw &\myul   &\qw &\mydl &\qw     &\qw &\mydl &                 &\qw &\mydl &\gate[2]{\phi_C'} &\qw &\mydl &                 &\qw &\mydl &                 &\qwrs{$c$}\\
\lstick{$g$}&                 &    &\myul   &\qw &      &\myuull &\qw &\myul &\qw              &\qw &\myul &                 &\qw &\myul &\qw              &\qw &\myul &                 &\qwrs{$g$}
\end{tikzcd}
\end{center}
It is a straightforward exercise to show that this circuit is
equivalent to the na\"ive construction of \cref{ex:extension} by using
well-known properties of swap operations.
\end{example}

Induced catalytic embeddings can be used to concatenate embedings of
gate sets. Suppose, for example, that $\gatesetF$, $\gatesetG$, and
$\gatesetH$ are three gate sets and that $(\phi,\projector): \gatesetF
\to \circuits(\gatesetG)$ and $(\phi',\projector'):\gatesetG \to
\circuits(\gatesetH)$ are catalytic embeddings. We can then define a
catalytic embedding $(\phi'',\projector''):\gatesetF \to
\circuits(\gatesetH)$ by setting
\[
(\phi'',\projector'') = \overline{(\phi',\projector')}\circ (\phi, \projector).
\]
Of course, the catalytic embedding $(\phi'',\projector'')$ itself
induces a catalytic embedding
$\overline{(\phi'',\projector'')}:\circuits(\gatesetF) \to
\circuits(\gatesetH)$.

\section{Linear Catalytic Embeddings}
\label{sec:lincats}

Catalytic embeddings, as introduced in the previous section, are very
general. In principle, a catalytic embedding can be defined for
arbitrary collections of circuits $\circuitsalt$ and $\circuitsalt'$
and is not required to preserve any structure, beyond that imposed by
the catalytic condition. This generality can make it rather daunting
to construct catalytic embeddings that might prove useful in any
way. We thus turn our attention to catalytic embeddings that preserve
the structure that may be present in underlying the collections of
matrices $\unitary(\circuitsalt)$ and $\unitary(\circuitsalt')$.

\begin{definition}
    \label{def:strong}
     We say that a catalytic embedding $(\phi, \projector):
     \circuitsalt\to \circuitsalt'$ is \emph{strong} when,
    \[
        e\circ\phi(C)=e\circ\phi(D)\iff e(C)=e(D)
    \]
    for all $C,D\in\circuitsalt$.
\end{definition}

\begin{proposition}
    \label{prop:strongequiv}
    Let $(\phi, \projector): \circuitsalt \to \circuitsalt'$ be a
    catalytic embedding. Then $(\phi, \projector)$ is strong if, and
    only if, there exists an injective function $\mu
    :\unitary(\circuitsalt) \to \unitary(\circuitsalt')$ such that the
    following diagram commutes:
    \begin{center}
    \begin{tikzcd}[row sep=large,column sep=large,nodes={inner sep=5pt}]
        \circuitsalt \ar[r, "\phi"]\ar[d, swap, "e"]& \circuitsalt'  \ar[d, "e"] \\
        \unitary[\circuitsalt] \ar[r, swap, "\mu"]& \unitary[\circuitsalt'].
    \end{tikzcd}
    \end{center}
    Moreover, the function $\mu :\unitary(\circuitsalt) \to
    \unitary(\circuitsalt')$ is uniquely determined by $\phi$.
\end{proposition}

\begin{proof}
  For the left-to-right direction, let
  $s:\unitary[\circuitsalt]\rightarrow\circuitsalt$ be any exact
  synthesis function and define $\mu=e\circ\phi\circ s$. Then, for any
  $C\in\circuitsalt$, we have $\mu\circ e(C) = e\circ\phi\circ s\circ
  e (C)$.  The function $s$ is an exact synthesis function, so that
  $e\circ s$ is trivial and therefore that $e\circ (s\circ e)(C) =
  (e\circ s)\circ e(C)= e(C)$. Because $(\phi,\projector)$ is strong,
  this implies that $e\circ\phi \circ s\circ e(C) = e\circ \phi (C)$
  and therefore that $\mu\circ e(C)= e\circ\phi(C)$, so that the
  diagram in \cref{prop:strongequiv} commutes. Reasoning similarly, we
  get
  \[
      \mu(U)=\mu(V)\implies e\circ\phi\circ s(U)=e\circ\phi\circ
      s(V)\implies e\circ s(U)=e\circ s(V)\implies U=V,
  \]
  so that $\mu$ is injective.
  
  For the right-to-left direction, assume that there exists an
  injective function $\mu :\unitary(\circuitsalt) \to
  \unitary(\circuitsalt')$ making the diagram in
  \cref{prop:strongequiv} commute. The equality $e(U)=e(V)$ implies
  $\mu\circ e (U) = \mu \circ e(V)$, which implies $e\circ \phi
  (U)=e\circ \phi(V)$ by commutativity of the diagram. Conversely, the
  equality $e\circ \phi (U) = e\circ \phi (V)$ implies $\mu\circ e(U)
  = \mu\circ e(V)$ by commutativity of the diagram, which implies
  $e(U)=e(V)$ by injectivity of $\mu$. Hence, $(\phi,\projector)$ is
  strong.

  To see that $\mu$ is uniquely defined, note that if
  $\mu':\unitary[\circuitsalt]\to\unitary[\circuitsalt']$ makes the
  diagram in \cref{prop:strongequiv} commute then, for any
  $U\in\unitary[\circuitsalt']$, we have $\mu'(U) = \mu' \circ (e\circ
  s)(U) = \mu' \circ e (s(U)) = e\circ \phi (s(U)) = \mu (U)$.
\end{proof}

Note that in defining the function $\mu$ in \cref{prop:strongequiv},
one can choose any exact synthesis function. Indeed, if $s$ and $s'$
are two such functions, then the uniqueness of $\mu$ implies that
$e\circ \phi \circ s = e\circ \phi \circ s'$.

In practical contexts, one typically studies collections of circuits
of the form $\circuits(\gatesetG)$ and $\circuits(\gatesetH)$, for
some gate sets $\gatesetG$ and $\gatesetH$ of interest. In these cases
the collections $\circuits(\gatesetG)$ and $\circuits(\gatesetH)$ are
closed under several circuit-building operations. Consequently, the
collections $\unitary[\circuits(\gatesetG)]$ and
$\unitary[\circuits(\gatesetH)]$ are endowed with some structure. The
next proposition shows that, for collections of this form, the map
$\mu$ associated with a strong induced catalytic embedding preserves
much of this structure.

\begin{proposition}
    \label{prop:strongembed}
    Let $\gatesetG$ and $\gatesetH$ be two gate sets, let $(\phi,
    \projector): \circuits(\gatesetG)\to \circuits(\gatesetH)$ be a
    strong induced $k$-dimensional catalytic embedding, and let $\mu
    :\unitary(\circuits(\gatesetG)) \to
    \unitary(\circuits(\gatesetH))$ be as in
    \cref{prop:strongequiv}. Then, for any
    $U,V\in\unitary[\circuits(\gatesetG)]$ with $\rank(U)=\ell$ and
    $\rank(V)=m$,
    \begin{enumerate}
        \item $ \mu(U)(I_\ell\otimes\projector)=U\otimes\projector$,
        \item $\mu(I_n)=I_{k n}~\text{for any positive integer}~n$,
        \item $\mu(VU)=\mu(V)\mu(U) ~\text{if}~\ell=m$, and
        \item $\mu(V\otimes U) =\left(\swapalt(\ell,m)\otimes
          I_k\right)\left(I_\ell\otimes\mu(V)\right)\left(\swapalt(m,\ell)\otimes
          I_k\right)\left(I_m\otimes\mu(U)\right)$.
    \end{enumerate}
\end{proposition}

\begin{proof}
  Recall from \cref{prop:strongequiv} that for any exact synthesis
  function $s:\unitary(\circuits(\gatesetH)) \to
  \circuits(\gatesetH)$, we have $\mu=e\circ\phi\circ s$. We take
  advantage of this here by choosing a convenient exact synthesis
  function for each one of the assertions to be established. For the
  first assertion, let $s$ be arbitrary. By the catalytic condition
  and the fact that $e\circ s$ is trivial, we have
  \[
      \mu(U)(I_\ell\otimes\projector) = e\circ\phi\circ
      s(U)(I_\ell\otimes\projector) = e\circ s(U) \otimes\projector=
      U\otimes\projector
  \]
  so that the assertion follows. For the second assertion, take an
  exact synthesis map $s$ such that $s(I_n)$ is the circuit
  $I_n$. Then, $\mu(I_n) = e\circ \phi \circ s(I_n) = e\circ \phi
  (I_n) = e( I_n \otimes I_k) = I_{nk}$ by
  \cref{def:inducedcatembedcirc}, since $(\phi, \projector)$ is an
  induced catalytic embedding. For the third assertion, let $s$ be a
  synthesis map such that $s(VU)=s(V)\circ s(U)$. Then, writing $\mu$
  as $\mu=e\circ \phi \circ s$, we get
    \begin{align*}
      \mu(VU)=e\circ\phi\circ s(VU) = e\circ \phi (s(V)\circ s(U)) =
      e((\phi\circ s(V))\circ(\phi\circ s(U))) = (e\circ\phi\circ
      s(V)) \circ(e\circ\phi\circ s(U)) = \mu(V)\mu(U),
  \end{align*}
  using \cref{def:evaluation,def:inducedcatembedcirc}. For the final
  assertion, note that $V\otimes U = (V\otimes I_\ell)(I_m\otimes
  U)=(I_m\otimes U)(V\otimes I_\ell)$. Thus, the fourth assertion
  follows from the third, as long as we can show that $\mu(I_m\otimes
  U)=I_m\otimes \mu(U)$ and that
  \[
  \mu(V\otimes I_\ell) = \left(\swapalt(\ell,m)\otimes
  I_k\right)\left(I_\ell\otimes\mu(V)\right)\left(\swapalt(m,\ell)\otimes
  I_k\right).
  \]
  To show that $\mu(I_m\otimes U)=I_m\otimes \mu(U)$, let $s$ be such
  that $s(I_m\otimes U) = I_m \otimes s(U)$. We then have, using
  \cref{def:evaluation,def:inducedcatembedcirc}, and properties of the
  unitary $\swapalt(n,k)$,
  \begin{align*}
  \mu(I_m\otimes U) &= e\circ \phi \circ s(I_m\otimes U) \\
  &= e\circ \phi (I_m\otimes s(U)) \\
  &= e(((I_m\otimes \phi(s(U))) 
  \circ (I_m\otimes \swapalt(k,n)))
  \circ (I_{mk} \otimes I_n))
  \circ (I_m \circ \swapalt(n,k))\\
  &= e(I_m\otimes \phi(s(U))) \\
  &= I_m\otimes e(\phi(s(U))) \\
  &= I_m\otimes \mu(U).  
  \end{align*}
  To show that $ \mu(V\otimes I_\ell) = \left(\swapalt(\ell,m)\otimes
  I_k\right)\left(I_\ell\otimes\mu(V)\right)\left(\swapalt(m,\ell)\otimes
  I_k\right)$, one can reason analogously, choosing an exact synthesis
  function $s$ satisfying $s(V\otimes I_\ell)=s(V)\otimes I_\ell$.
\end{proof}

\begin{remark}
  Both $\unitary[\circuits(\gatesetG)]$ and
  $\unitary[\circuits(\gatesetH)]$ have the structure of a
  groupoid. \cref{prop:strongembed} shows that the map $\mu$ is a
  faithful groupoid functor from $\unitary[\circuits(\gatesetG)]$ to
  $\unitary[\circuits(\gatesetH)]$.
\end{remark}

It is natural to seek a converse statement to
\cref{prop:strongembed}. Indeed, this would provide conditions under
which a function $\unitary[\circuits(\gatesetG)] \to
\unitary[\circuits(\gatesetH)]$ can be used to define a catalytic
embedding $\circuits(\gatesetG) \to \circuits(\gatesetH)$. This is the
goal of the following proposition.

\begin{proposition}
  \label{prop:strongunitarymap}
  Let $\gatesetG$ and $\gatesetH$ be two gate sets and let
  $s:\unitary(\circuits(\gatesetH)) \to \circuits(\gatesetH)$ be an
  exact synthesis function. If
  $\projector:\hilbert_k\rightarrow\hilbert_k$ is an orthogonal
  projector of nonzero rank and $\mu:\unitary(\circuits(\gatesetG))
  \rightarrow \unitary(\circuits(\gatesetH))$ is a function such that,
  for any $U,V\in\unitary[\circuits(\gatesetG)]$ with $\rank(U)=\ell$
  and $\rank(V)=m$,
  \begin{enumerate}
    \item $\mu(U)(I_\ell\otimes\projector)=U\otimes\projector$,
    \item $\mu(I_n)=I_{k n}~\text{for any positive integer } n$,
    \item $\mu(VU)=\mu(V)\mu(U) ~\text{if}~\ell=m$, and
    \item $\mu(V\otimes U)=\left(\swapalt(\ell,m)\otimes
      I_k\right)\left(I_\ell\otimes\mu(V)\right)\left(\swapalt(m,\ell)\otimes
      I_k\right)\left(I_m\otimes\mu(U)\right)$,
  \end{enumerate}
  then $\overline{(\phi,\projector)}$, where $\phi=s\circ \mu \circ
  e$, is a strong and homogeneous $k$-dimensional catalytic embedding
  of $\circuits(\gatesetG)$ in $\circuits(\gatesetH)$.
\end{proposition}

\begin{proof}
    Write $\phi=s\circ \mu \circ e$. For $G\in\gatesetG$, we have
    \[
        e(\phi (G))(I\otimes \projector) = e\circ s\circ\mu\circ
        e(G)(I\otimes\projector) = \mu(e(G))(I\otimes\projector) =
        e(G)\otimes\projector
    \]
    so that $(\phi,\projector):\gatesetG \to \circuits(\gatesetH)$ is
    a homogeneous catalytic embedding of dimension $k$. Write
    $\overline{(\phi,\projector)}$ for the induced catalytic
    embedding. To show that $\overline{(\phi, \projector)}$ is strong,
    we use \cref{prop:strongequiv}. By the condition 1, $\mu$ is
    injective so that we only need to show that the diagram of
    \cref{prop:strongequiv} commutes. Let
    $U\in\circuits(\gatesetG)$. By \cref{def:circs}, $U$ is a
    well-formed word over $\gatesetG\cup \s{I_n \mid n\in\N^*} \cup
    \s{\swapalt(m,n) \mid m,n\in\N^*} \cup \s{\circ, \otimes, (,)}$
    and $\overline{\phi}(U)$ is the word obtained from $U$ by
    \cref{def:inducedcatembedcirc}. Note that, for $G\in\gatesetG$, we
    have
    \[
        e\circ\overline{\phi}(G) = e\circ \phi(G) = e\circ s\circ\mu\circ e(G) =
        \mu\circ e(G),
    \]
    so that $e\circ\overline{\phi}$ and $\mu\circ e$ agree on
    $\gatesetG$. Since $e$ is multiplicative and behaves trivially
    with respect to identities and swaps, conditions 2--4 ensure that
    \[
        e\circ\overline{\phi}(U) = \mu\circ e(U).
    \]
    Hence, the diagram commutes as desired.
\end{proof}

\begin{definition}[Lifting]
    \label{def:lifting}
    Let $\mu$ and $\phi$ be as in \cref{prop:strongunitarymap}. The
    catalytic embedding $\phi$ is called the \emph{lifting} of $\mu$.
\end{definition}

\cref{prop:strongunitarymap} suggests that to find a catalytic
embedding between $\circuits(\gatesetG)$ and $\circuits(\gatesetH)$,
we can turn to the corresponding collection of matrices.  Note that
for any gate set $\gatesetG$, we can always find a number ring
$\ringS$ such that $\unitary[\gatesetG]\subseteq\unitary[\ringS]$. In
fact, as discussed earlier, some important gate sets in quantum
computing can be identified as exactly the set of unitary matrices of
appropriate size which have entries in a particular number ring. These
gate sets, such as the Clifford+$T$ gate set, are often used for the
fault-tolerant implementation of unitary operations. These properties
naturally motivate the study of maps between unitary matrices over
number rings, paying particular attention to maps into number rings
associated with fault-tolerant gate sets. We thus consider what
additional structure can be added to facilitate studying such
maps. Unitary groups over rings are only closed under multiplication.
However, when representing the elements of these groups in the
computational basis as matrices over number rings, it is often
convenient to use the additive structure of the underlying algebra of
matrices.  We thus focus on maps which respect this additive structure
on matrices whilst also preserving unitarity.

\begin{proposition}
    \label{prop:embedformal}
    Let $\ringR$ and $\ringS$ be Kroneckerian number rings with
    $\fieldfrac{\ringR},\fieldfrac{\ringS}\subseteq \fieldK$ for a
    field $\fieldK$, $\ringT=\ringR\bigcap\ringS$,
    $A,B\in\matrices[\ringS]$, and $C\in\matrices[\ringT]$. Suppose
    that $\Phi:\matrices[\ringS]\rightarrow\matrices[\ringR]$
    satisfies the following conditions:
    \begin{enumerate}
        \item $\Phi(AB) = \Phi(A)\Phi(B)$ when $AB$ is defined
        \item $\Phi(A+B)=\Phi(A)+\Phi(B)$ when $A+B$ is defined
        \item $\Phi(I_n)=I_{kn}$ for some fixed $k$
        \item $\Phi(A^\dagger)=\Phi(A)^\dagger$
        \item $\Phi(C\otimes A)= C\otimes\Phi(A)$
        \item There exists an orthogonal projector
          $\projector:\hilbert_k\rightarrow\hilbert_k$ with
          $\rank(\projector)>0$ such that
          $\Phi(M)(I\otimes\projector)=M\otimes\projector$ where
          $\dim(M)=\dim(I)$ for all $M\in\matrices[\ringS]$
    \end{enumerate}
    Then $\Phi(\unitary[\ringS])\subset\unitary[\ringR]$ and the map
    $\mu:\unitary[\ringS]\rightarrow\unitary[\ringR]$ with
    $\mu(U):=\Phi(U)$ can be lifted to a homogeneous, strong catalytic
    embedding of $\circuits(\ringS)$ in $\circuits(\ringR)$ of
    dimension $k$.
\end{proposition}

\begin{proof}
    Let $U\in\unitary[\ringS]$. By conditions 3 and 4 we have
    \[
    \Phi(U)\Phi(U)^\dagger=\Phi(U)\Phi(U^\dagger)=\Phi(UU^\dagger)=\Phi(I)=I_k,
    \]
    and so $\Phi$ maps unitary matrices to unitary matrices and $\mu$
    is well-defined. Next, we would like to show that $\mu$ satisfies
    the conditions presented in \cref{prop:strongunitarymap} so that
    we can lift it to a strong and homogeneous catalytic embedding. By
    inspection, we see that the only condition not immediately
    satisfied by our assumptions is the action of $\mu$ on $V\otimes
    U$ for $U,V\in\unitary[\ringS]$. For $\rank(U)=\ell$ and
    $\rank(V)=m$, we calculate
    \begin{align*}
        \mu(V\otimes U) &= \mu(\swapalt(l,m)\circ (I_\ell\otimes V)\circ\swapalt(m,l)\circ (I_m\otimes U))\\
        &= \mu(\swapalt(l,m))\mu(I_\ell\otimes V)\mu(\swapalt(m,l))\mu(I_m\otimes U)\\
        &= \mu(\swapalt(l,m))\otimes I_1)\mu(I_\ell\otimes V)\mu(\swapalt(m,l)\otimes I_1)\mu(I_m\otimes U)\\
        &= (\swapalt(l,m))\otimes \mu(I_1))(I_\ell\otimes \mu(V))(\swapalt(m,l)\otimes \mu(I_1))(I_m\otimes \mu(U))\\
        &= (\swapalt(l,m))\otimes I_k)(I_\ell\otimes \mu(V))(\swapalt(m,l)\otimes I_k)(I_m\otimes \mu(U))
    \end{align*}
    and thus all conditions of \cref{prop:strongunitarymap}
    hold. Therefore, we can lift $\mu$ to a homogeneous and strong
    catalytic embedding of $\circuits(\ringS)$ over
    $\circuits(\ringR)$ of dimension $k$.
\end{proof}

As it might not be immediately obvious from the proof of
\cref{prop:embedformal}, it is worth considering one consequence of
condition 5. Let $U,V\in\circuits(\ringS)$ be such that
$e(V)\in\unitary[\ringR]$ with $\dim(V)=\ell$ and $\dim(U)=m$. Lifting
$\mu$ as in \cref{prop:strongunitarymap} to $\phi$, we get
    \begin{align*}
        e\circ\phi(V\otimes U)&=\left(\swapalt(\ell,m)\otimes I_k\right)\left(I_\ell\otimes e\circ\phi(V)\right)\left(\swapalt(m,\ell)\otimes I_k\right)\left(I_m\otimes e\circ\phi(U)\right)\\
        &=\left(\swapalt(\ell,m)\otimes I_k\right)\left(I_\ell\otimes \mu\circ e(V)\right)\left(\swapalt(m,\ell)\otimes I_k\right)\left(I_m\otimes \mu\circ e(U)\right)\\
        &=\mu(e(V)\otimes e(U))\\
        &=\Phi(e(V)\otimes e(U))\\
        &=e(V)\otimes\Phi(e(U))\\
        &=e(V)\otimes \mu\circ e(U)\\
        &=e(V)\otimes e\circ\phi(U).
    \end{align*}
The above derivation shows that any lifting of $\mu$ acts trivially on
the elements of $\circuits(\ringR)$ which are equivalent to an element
of $\circuits(\ringS)$. This is consistent with the type of action we
might desire in practice: we have no obvious need to embed gates or
circuits which we already have direct access to.

We focus on studying maps that satisfy the conditions in
\cref{prop:embedformal} for the remainder of this document and so we
present the following definitions.

\begin{definition}[Linear catalytic embedding, Pre-embedding]
  \label{def:embedformal}
  Let $\Phi$ and $\mu$ be defined as in \cref{prop:embedformal} and
  let $\phi$ be the catalytic embedding that results from lifting
  $\mu$. We say that $\phi$ is a \emph{linear} catalytic embedding and
  that $\Phi$ is the \emph{pre-embedding} of $\phi$.
\end{definition}

While we have introduced embeddings with varying amounts of structure,
we might wonder whether such constructions exist in practice. In fact,
there are distinct instances of each type of embedding, as the
following example highlights.

\begin{example}
    \label{ex:linembed1}
    Let $\alpha$ be a primitive third root of unity. Consider the gate set
    \[
        \gatesetG=\s{R,X}\quad\text{with}\quad e(R)=\begin{bmatrix} 1 & 0 \\ 0 & \alpha\end{bmatrix},~ e(X)=\begin{bmatrix} 0 & 1 \\ 1 & 0\end{bmatrix}.
    \]
    Below are different examples of catalytic embeddings of $R$ and $X$ in $\circuits(\Q)$.
    \begin{enumerate}
        \item $e(\phi_R)=\begin{bmatrix}
            I_3 & 0\\
            0 & \Lambda
        \end{bmatrix}$, $e(\phi_X)=\begin{bmatrix}
            0 & I_3\\
            I_3 & 0
        \end{bmatrix}$, $\projector= \frac{1}{3}\begin{bmatrix}
            1 & \alpha & \alpha^2\\
            \alpha^2 & 1 & \alpha\\
            \alpha & \alpha^2 & 1
        \end{bmatrix}$,  for $\Lambda=\frac{1}{3}\begin{bmatrix}
            -2 & -2 & 1\\
            1 & -2 & -2\\
            -2 & 1 & -2
        \end{bmatrix}$
        
        \item $e(\phi_R)=\begin{bmatrix}
            I_3 & 0\\
            0 & \Lambda
        \end{bmatrix}$, $e(\phi_X)=\begin{bmatrix}
            0 & I_3\\
            I_3 & 0
        \end{bmatrix}$, $\projector= \frac{1}{3}\begin{bmatrix}
            1 & \alpha & \alpha^2\\
            \alpha^2 & 1 & \alpha\\
            \alpha & \alpha^2 & 1
        \end{bmatrix}$,  for $\Lambda=\begin{bmatrix}
            0 & 0 & 1\\
            1 & 0 & 0\\
            0 & 1 & 0
        \end{bmatrix}$
        
        \item $e(\phi_R)=\begin{bmatrix}
            I_4 & 0\\
            0 & \Lambda
        \end{bmatrix}$, $e(\phi_X) = \begin{bmatrix}
            0 & I_4 \\
            I_4 & 0
        \end{bmatrix}$, $\projector = \frac{1}{6}\begin{bmatrix}
            3 & 1+2\alpha & 1+2\alpha & 1+2\alpha\\
            1+2\alpha^2 & 3 & 1+2\alpha & 1+2\alpha^2\\
            1+2\alpha^2 & 1+2\alpha^2 & 3 & 1+2\alpha\\
            1+2\alpha^2 & 1+2\alpha & 1+2\alpha^2 & 3
        \end{bmatrix}$,  for $\Lambda=\frac{1}{2}\begin{bmatrix}
            -1 & -1 & -1 & -1\\
            1 & -1 & -1 & 1\\
            1 & 1 & -1 & -1\\
            1 & -1 & 1 & -1
        \end{bmatrix}$
    \end{enumerate}
    Each of these embeddings is homogeneous, and induces a catalytic
    embedding $\overline{(\phi,\projector)}$ from
    $\circuits(\gatesetG)$ to $\circuits(\Q)$. These induced
    embeddings highlight the different properties a catayltic
    embedding may have:
    \begin{itemize}
        \item Embedding 1 is not strong: $e(R\circ R\circ R) = e(X\circ X)$ while $e(\phi_R\circ\phi_R\circ\phi_R)\neq e(\phi_X\circ\phi_X)$.
        \item Embedding 2 is strong, but not linear: $e(R\circ X\circ R)+e(R\circ R\circ X\circ R\circ R ) = -e(X)$ while \\
        $e(\phi_R\circ \phi_X\circ \phi_R)+e(\phi_R\circ\phi_R\circ \phi_X\circ \phi_R\circ \phi_R) \neq -e(\phi_X)$
        \item Embedding 3 is linear.
    \end{itemize}
    These statements can be verified by direct computation.
\end{example}

We conclude this section with a number of simple properties of
pre-embeddings and linear catalytic embeddings. First, we start with
two lemmas that aid in characterizing possible constructions for
linear catalytic embeddings.

\begin{lemma}
  \label{lem:middle}
  Let $\ringR$ and $\ringS\subset\ringT$ be Kroneckerian subrings of a
  number field. If
  $\Phi:\matrices[\ringT]\rightarrow\matrices[\ringR]$ is the
  pre-embedding for a linear catalytic embedding, then
  $\Phi|_{\ringS}:\matrices[\ringS]\rightarrow\matrices[\ringR]$ is
  the pre-embedding for a linear catalytic embedding.
\end{lemma}

While \cref{lem:middle} may appear rather trivial, it is nonetheless
quite useful for proving the (non-)existence of certain linear
catalytic embeddings. For example, the contrapositive of
\cref{lem:middle} implies that if no such $\Phi|_{\ringS}$ can be
constructed, then $\Phi$ cannot exist. This is especially powerful in
instances where the elements of $\ringT$ are simpler to describe than
those of $\ringS$ (or vice-versa).

\begin{lemma}
    \label{lem:nonewndenoms}
    Let $\ringR$ and $\ringS$ be Kroneckerian subrings of a number
    field. If a pre-embedding
    $\Phi:\matrices[\ringS]\rightarrow\matrices[\ringR]$ for a linear
    catalytic embedding exists, then $\ringS\cap\ringR = \ringS\cap
    \fieldfrac{\ringR}$.
\end{lemma}

\begin{proof}
    We have $\ringS\cap\ringR\subseteq\ringS\cap\fieldfrac{\ringR}$,
    and so we just need to show that there is no
    $s\in\ringS\cap\fieldfrac{\ringR}$ with
    $s\not\in\ringS\cap\ringR$. Suppose that such an $s$
    existed. There always exists some nonzero $t\in\ringS\cap\ringR$
    such that $ts\in\ringS\cap\ringR$, which simultaneously implies
    \[
        \Phi(ts) = t\Phi(s)\quad\text{and}\quad \Phi(ts) = ts I.
    \]
    As $\ringR$ is an integral domain, by assumption we would conclude
    that $\Phi(s)=s I$, which is absurd as
    $s\in\ringS\cap\fieldfrac{\ringR}$ but $s\not\in\ringS\cap\ringR$
    implies $s\not\in\ringR$. Thus, there can be no such $s$.
\end{proof}

\cref{lem:nonewndenoms} precludes ``reducing'' the set of denominators
that appear in number rings via embedding. This has implications for
embedding to restricted gate sets such as Clifford+T, whose
corresponding unitaries only permit particular denominators. We end
this section with two important results for the implementation of
linear catalytic embeddings. First, we show that linear catalytic
embeddings respect the direct sum operation.

\begin{proposition}
    \label{prop:directsum}
    Let $\ringR$ and $\ringS$ be Kroneckarian number rings and
    $\Phi:\matrices[\ringS]\rightarrow\matrices[\ringR]$ be the
    pre-embeddding of a linear catalytic embedding. Then, for
    $A,B\in\matrices[\ringS]$,
    \[
        \Phi(A\oplus B) = \Phi(A)\oplus\Phi(B).
    \]
\end{proposition}
\begin{proof}
    Let $E_{ij}$ the standard basis matrices for
    $\matrices(\ringS)$. We can write
    \[
    A = \sum_{i,j=1}^k A_{ij}E_{ij} \text{ and } B = \sum_{i,j=1}^m B_{ij}E_{ij}
    \] for $k$ and $m$ the dimension of $A$ and $B$ respectively. Then 
    \[
    A\oplus B = \sum_{i,j=1}^k A_{ij}E_{ij} + \sum_{i,j=k+1}^{m+k} B_{ij}E_{ij}\] and \[\Phi(A\oplus B) = \sum_{i,j=1}^k E_{ij}\otimes\Phi(A_{ij}) + \sum_{i,j=k+1}^{m+k} E_{ij}\otimes \Phi(B_{ij}).
    \] We also have 
    \[\Phi(A) = \sum_{i,j=1}^k E_{ij}\otimes\Phi(A_{ij}) \text{ and } \Phi(B) = \sum_{i,j=1}^{m} E_{ij}\otimes \Phi(B_{ij}).
    \] and so
    \begin{equation*}
    \Phi(A)\oplus \Phi(B) = \sum_{i,j=1}^k E_{ij}\otimes\Phi(A_{ij}) + \sum_{i,j=k+1}^{m+k} E_{ij}\otimes \Phi(B_{ij}) = \Phi(A\oplus B).\qedhere
    \end{equation*}
\end{proof}

\cref{prop:directsum} is useful when considering the action of linear
catalytic embeddings with respect to controlled operations. When
lifting a pre-embedding to a linear catalytic embedding $\phi$, we can
always construct a lift that preserves the controlled unitary
structure between the control and target registers. In other words,
for a gate $G$ which applies $U$ to a target register conditioned on a
control register being in state $\ket{a}$, we can lift to the gate
$\phi(G)$ which applies $\phi(U)$ to the target and catalyst registers
conditioned on the control register being in state $\ket{a}$.

Finally, we show that the concatenation of two linear catalytic
embeddings is a linear catalytic embedding, under a minor assumption.

\begin{proposition}
    \label{prop:linearconcat}
    Let $\ringR$, $\ringS$, and $\ringT$ be Kroneckarian subrings of a
    number field with $\ringT\cap\ringR\subseteq\ringS$ and
    $\Phi_1:\matrices[\ringT]\rightarrow\matrices[\ringS]$ and
    $\Phi_2:\matrices[\ringS]\rightarrow\matrices[\ringR]$ be
    pre-embeddings for linear catalytic embeddings. Then the
    concatenation $\Phi_2\circ\Phi_1$ is a pre-embedding for a linear
    catalytic embedding.
\end{proposition}

\begin{proof}
    We show that the composition $\Phi_2\circ\Phi_1$ satisfies the
    conditions of a pre-embedding.
    \begin{itemize}
        \item[1., 2.] Properties 1 and 2 follow from the composition of ring homomorphism being a homomorphism.
        
        \item[3.] Let $k_1$ and $k_2$ be such that $\Phi_1(I_n)=I_{k_1n}$ and $\Phi_2(I_n)=I_{k_2n}$ for all $n$. Then $\Phi_2\circ\Phi_1(I_n)=\Phi_2(I_{k_1n})=I_{k_1 k_2n}$ for all $n$.
        
        \item[4.] Let $A\in\matrices(\ringS)$. Then $\Phi_2\circ\Phi_1(A^\dagger)=\Phi_2(\Phi_1(A)^\dagger)=(\Phi_2\circ\Phi_1(A))^\dagger$.

        \item[5.] Let $C\in\matrices(\ringT\cap\ringR)$. Then $C\in\matrices(\ringT\cap\ringS)$ and $C\in\matrices(\ringS\cap\ringR)$ as $\ringT\cap\ringR\subset\ringS$ implies $C\in\matrices(\ringR),\matrices(\ringS),\matrices(\ringT)$. So 
        \[
        \Phi_2\circ\Phi_1(C\otimes A)=\Phi_2(C\otimes\Phi_1(A))=C\otimes \Phi_2\circ\Phi_1(A).
        \]

        \item[6.] Let $\projector_1$ and $\projector_2$ be the projectors that satisfy the catalytic condition for $\Phi_1$ and $\Phi_2$ with dimension $k_1$ and $k_2$ respectively. Then $\Pi_1\otimes\Pi_2:\hilbert_{k_1 k_2}\rightarrow\hilbert_{k_1 k_2}$ with $\rank(\Pi_1\otimes\Pi_2)=\rank(\Pi_1)\rank(\Pi_2)>0$, and
        \[
            (\Pi_1\otimes\Pi_2)^\dag(\Pi_1\otimes \Pi_2)=(\Pi_1^\dag\otimes\Pi_2^\dag)(\Pi_1\otimes \Pi_2) = (\Pi_1^\dag\Pi_1)\otimes(\Pi_2^\dag\Pi_2) = \Pi_1\otimes\Pi_2
        \]
        so that $\Pi_1\otimes\Pi_2$ is a nonzero orthogonal projector on $\hilbert_{k_1 k_2}$. Checking that it satisfies the catalytic condition for $A\in\matrices(\ringT)$, we have
        \begin{align*}
            \Phi_2\circ\Phi_1(A)(I\otimes(\Pi_1\otimes \Pi_2))
            &=\Phi_2(\Phi_1(A))((I\otimes I_{k_1})\otimes\Pi_2)((I\otimes\Pi_1)\otimes I_{k_2})\\
            &=(\Phi_1(A)\otimes\Pi_2)((I\otimes\Pi_1)\otimes I_{k_2})\\
            &= (A\otimes\Pi_1)\otimes\Pi_2\\
            &= A\otimes(\Pi_1\otimes\Pi_2).\qedhere
        \end{align*}
    \end{itemize}
\end{proof}

\cref{prop:linearconcat} allows us to concatenate sequences of linear
catalytic embeddings together to simplify circuits. We will give an
example of this process in the following section.

\section{Properties of Linear Catalytic Embeddings}
\label{sec:props}

In this section, we establish several important properties of linear
catalytic embeddings. We show that linear catalytic embeddings can
often be thought of as linear representations. We also prove that a
linear catalytic embedding can be associated with a family of
catalysts which can be used to perform non-trivial transformations of
circuits. Finally, we provide a few simple examples to highlight the
construction of linear catalytic embeddings.

Recall that if $\ringT$ is a subring of $\ringS$, then $\ringS$ is a
$\ringT$-module, and that a generating set for $\ringS$ over $\ringT$
is a set $\Gamma \subseteq\ringS$ such that every element of $\ringS$
can be written as a finite $\ringT$-linear combination of elements of
$\Gamma$. We begin this section by showing that a linear catalytic
embedding is determined by its action on a generating set.

\begin{proposition}
    \label{prop:onlyS}
    Let $\ringR$ and $\ringS$ be Kroneckarian number rings with
    $\ringT = \ringR\cap\ringS$ and let
    $\Phi:\matrices[\ringS]\rightarrow \matrices[\ringR]$ be a
    pre-embedding. If $\Gamma$ is a generating set for $\ringS$ over
    $\ringT$, then the action of $\Phi$ on $\matrices[\ringS]$ is
    completely determined by its action on $\Gamma$.
\end{proposition}

\begin{proof}
    Let $\Gamma$ be a generating set for $\ringS$ over $\ringT$. Since
    $\Gamma$ generates $\ringS$ as a $\ringT$-module, for $s\in\ringS$
    we can write
    \[
        s = \sum_{g\in\Gamma}t^{(g)}g,
    \]
    for $t^{(g)}\in\ringT$. For each $i,j$, let
    $E_{ij}\in\matrices[\ringS]$ be the standard basis matrix, i.e.,
    the matrix whose $(i,j)$ entry is 1 and whose other entries are
    all 0. Suppose that $M\in\matrices[\ringS]$. Then, for each $i,j$,
    we can write the $i,j$-th entry of $M$ as
    \[
    M_{i,j}=\sum_{g\in\Gamma} t_{i,j}^{(g)}g.
    \] 
    By applying the properties of a pre-embedding and the fact that
    scalar multiplication is equivalent to taking the tensor product
    with a $1\times1$ matrix, we then have
    \[
        \Phi(M)  = \Phi\left(\sum_{i,j} M_{i,j} E_{i,j} \right) = \Phi\left(\sum_{i,j}\sum_{g\in\Gamma} t_{i,j}^{(g)}g\cdot E_{i,j}\right) = \sum_{i,j}\sum_{g\in\Gamma}\Phi\left(t_{i,j}^{(g)}E_{i,j}\otimes g\right) = \sum_{i,j}\sum_{g\in\Gamma} t_{i,j}^{(g)}E_{i,j}\otimes \Phi(g).
    \]
    Hence, for all $M\in\matrices[\ringS]$, $\Phi(M)$ is completely
    determined by the action of $\Phi$ on $\Gamma$.
\end{proof}

In much the same way that linear operators are determined by their
action on a basis, linear catalytic embeddings are determined by their
action on a small set that generates $\matrices(\ringS)$. However,
unlike linear operators, linear catalytic embeddings are not free to
act in any way on a generating set. The set of matrices
$\matrices(\ringS)$ has more structure than than that of a module and
linear catalytic embeddings have to preserve this additional
structure. Our next theorem addresses this additional structure and
shows that, in many important cases, a linear catalytic embedding is
completely determined by a single element. In proving this result we
use the \emph{the Galois closure} of a field extension
$\fieldF\backslash \fieldE$, which is the smallest field in which
irreducible polynomials in $\fieldE$ with a linear factor in $\fieldF$
can be factored completely. For further details, we direct the
interested reader to \cite{dummit2004abstract}.

\begin{theorem}
    \label{thm:problem}
    Let $\ringR$ and $\ringS$ be Kroneckarian number rings with
    $\ringT=\ringR\cap\ringS$. Let $\alpha\in\ringS$ such that
    $\fieldfrac{\ringS}=\fieldfrac{\ringT}[\alpha]$ with minimal
    polynomial $p\in\fieldfrac{\ringT}[x]$ over
    $\fieldfrac{\ringT}$. We can construct the pre-embedding of a
    linear catalytic embedding of $\circuits(\ringS)$ in
    $\circuits(\ringR)$ if and only if we can construct
    $\Lambda\in\matrices(\ringR)$ satisfying the following properties:
    \begin{enumerate}
        \item $\Lambda$ is normal
        \item $p(\Lambda)=0$
        \item $\Lambda$ has $\alpha$ as one of its eigenvalues
        \item there exists a generating set $\Gamma$ for $\ringS$ over
          $\ringT$ such that for every $g\in\Gamma$, written uniquely
          as the sum over powers of $\alpha$ as
        \[
            g = \sum_{j=0}^{d-1} c_j \alpha^j
        \]
        for some $c_j\in\fieldfrac{\ringT}$, the matrix
        \[
            \sum_{j=0}^{d-1} c_j \Lambda^j
        \]
        is a matrix over $\ringR$.
    \end{enumerate}
\end{theorem}
\begin{proof}    
    $(\Rightarrow)$ Let $\Phi:\matrices(\ringS) \rightarrow
  \matrices(\ringR)$ be a pre-embedding of a linear catalytic
  embedding. Let $\Lambda=\Phi(\alpha)$.
    \begin{enumerate}
    \item By the properties of linear catalytic embeddings and
      commutativity of $\ringS$,
    \[
        \Phi(\alpha)\Phi(\alpha)^\dagger = \Phi(\alpha)\Phi(\alpha^\dagger)=\Phi(\alpha \alpha^\dagger)=\Phi(\alpha^\dagger \alpha)=\Phi(\alpha^\dagger)\Phi(\alpha)=\Phi(\alpha)^\dagger\Phi(\alpha)
    \]
    and so $\Lambda$ is a normal matrix. 
    
    \item Let $p\in\fieldfrac{\ringT}[x]$ be the minimal (monic)
      polynomial of $\alpha$ over $\fieldfrac{\ringT}$. There exists
      nonzero $t\in\ringT$ such that $t\cdot p = q\in\ringT[x]$ by
      clearing denominators. By the properties of linear catalytic
      embeddings
    \[
    0 = \Phi(0) = \Phi(q(\alpha)) = q(\Phi(\alpha)) = t\cdot p(\Phi(\alpha)).
    \] As $\fieldfrac{\ringT}$ is an integral domain, $p(\Lambda)=0$. 
    
    \item Because $\Phi$ is the pre-embedding of a linear catalytic
      embedding, there exists a nonzero projector $\Pi$ such that
      $\Phi(\alpha)(I\otimes\Pi)=\alpha\otimes\Pi$. Let $v\neq 0$ be
      in the image of $\Pi$. We have
    \begin{align*}
        \Phi(\alpha)(I\otimes\Pi)v &= (\alpha\otimes\Pi)v\\
        \Phi(\alpha)v & =  \alpha v\\ 
        \Lambda v &= \alpha v.
    \end{align*} Thus, $\alpha$ is an eigenvalue of $\Lambda$.

    \item Let $\Gamma$ be a generating set for $\ringS$ over $\ringT$
      and $g\in\Gamma$. There exist $c_j\in\fieldfrac{\ringT}$ such
      that
      \[ g = \sum_{j=0}^{d-1} c_j \alpha^j.
      \]
      There exists nonzero $t\in\ringT$ such that $t\cdot c_j\in
      \ringT$ for all $c_j$ by clearing denominators. By the
      properties of linear catalytic embeddings and by
      $\fieldfrac{\ringT}$ an integral domain, we have \begin{align*}
        tg &= \sum_{j=0}^{d-1} tc_j \alpha^j\\ \Phi(tg) &=
        \Phi\left(\sum_{j=0}^{d-1} tc_j \alpha^j\right)\\ t\Phi(g) &=
        \sum_{j=0}^{d-1} tc_j \Phi(\alpha)^j\\ \Phi(g) &=
        \sum_{j=0}^{d-1} c_j \Phi(\alpha)^j\\ &= \sum_{j=0}^{d-1} c_j
        \Lambda^j\in \matrices(\ringR)
        \end{align*} because $\Phi(g)\in\matrices(\ringR)$.
    \end{enumerate}

    $(\Leftarrow)$ Suppose there exists $\Lambda$ that satisfies the
    properties listed above and let $\Gamma$ be a generating set for
    $\ringS$ over $\ringT$ for which property 4 is satisfied. As
    $\Lambda$ is such that $p(\Lambda)=0$, all eigenvalues of
    $\Lambda$ are roots of $p$. We define the function
    $\Phi:\matrices[\ringS]\rightarrow\matrices[\ringR]$ as follows:
    \begin{align*}
    \Phi(g) &= \sum_{j=0}^{d-1} c_j^{(g)} \Lambda^j \text{ where } g = \sum_{j=0}^{d-1} c_j^{(g)} \alpha^j\text{ for all }g\in\Gamma\text{ and}\\
    \Phi(M) &= \sum_{i,j}\sum_{g\in\Gamma} t_{i,j}^{(g)}E_{i,j}\otimes \Phi(g) \text{ where } M = \sum_{i,j}\sum_{g\in\Gamma} t_{i,j}^{(g)}E_{i,j}\otimes g.
    \end{align*}
    The function $\Phi$ is well-defined as each element of $\ringS$
    can be written uniquely as a linear combination of powers of
    $\alpha$ over $\fieldfrac{\ringT}$. By extension, we can also
    uniquely write each element of $\matrices[\ringS][n]$ as a linear
    combination of powers of $\alpha$ over
    $\matrices[\fieldfrac{\ringT}][n]$. Thus for arbitrary
    $M\in\matrices[\ringS][n]$, we can always write
        \[
            \Phi:M = \sum_{j=0}^{d-1}M_j \alpha^j \mapsto
            \sum_{j=0}^{d-1}M_j \otimes \Lambda^j
        \]
    for some unique $M_j\in\matrices[\fieldfrac{\ringT}][n]$. We show
    that $\Phi$ satisfies the properties of the pre-embedding of a
    linear catalytic embedding. Let $A,B\in\matrices(\ringS)$, and
    $C\in\matrices(\ringT)$.

    \begin{enumerate}
        \item The space $\hilbert_k$ is spanned by any linearly
          independent choice of eigenvectors of $\Lambda$. Since
          $p(\Lambda)=0$, each such eigenvector has a corresponding
          eigenvalue $\sigma(\alpha)$ for $\sigma$ a
          $\fieldfrac{\ringT}$-fixing automorphism of the Galois
          closure of $\fieldfrac{\ringS}$ as an extension of
          $\fieldfrac{\ringT}$. For arbitrary $\ket{v}\in\hilbert_n$
          and eigenvector $\ket{\sigma(\alpha)}$ of $\Lambda$, we have
        \begin{align*}
            \Phi(M) (\ket{v}\otimes\ket{\sigma(\alpha)}) = \sum_{j=0}^{d-1}M_j\ket{v} \otimes \Lambda^j\ket{\sigma(\alpha)}
            =\sum_{j=0}^{d-1}M_j\ket{v} \otimes \sigma(\alpha)^j\ket{\sigma(\alpha)}
            =\sigma(M)\ket{v}\otimes\ket{\sigma(\alpha)}.
        \end{align*}
        Because $\sigma$ is a homomorphism and vectors of the form
        $\ket{v}\otimes\ket{\sigma(\alpha)}$ span $\hilbert_{kn}$, we
        conclude $\Phi(AB)=\Phi(A)\Phi(B)$ when $AB$ is defined.
        
        \item By definition of $\Phi$ and linearity of the tensor
          product, $\Phi(A+B)=\Phi(A)+\Phi(B)$ when $A+B$ is defined.
        
        \item $I_n$ has the unique decomposition as a linear
          combination of powers of $\alpha$ over $\fieldfrac{\ringT}$
          of $I_n\alpha^0$. Thus
        \[
            \Phi(I_n) = I_n\otimes\Lambda^0 = I_n\otimes I_k = I_{kn}.
        \]
        
        \item For any $M\in\matrices[\ringS][n]$ we write
          $M=\sum_{j=0}^{d-1}M_j\alpha^j$ for unique
          $M_j\in\matrices[\fieldfrac{\ringT}][n]$. This implies
        \[
            M^\dagger = \sum_{j=0}^{d-1}M_j^\dagger(\alpha^\dagger)^j = \sum_{j=0}^{d-1}N_j \alpha^j
        \]
        again for some unique $N_j\in\matrices[\fieldfrac{\ringT}][n]$
        since $M^\dagger\in\matrices[\ringS][n]$ by $\ringS$
        Kroneckerian. Consider eigenvector $\ket{\sigma(\alpha)}$ of
        $\Lambda$ for $\sigma$ a $\fieldfrac{\ringT}$-fixing
        automorphism of the Galois closure of $\fieldfrac{\ringS}$ as
        an extension of $\fieldfrac{\ringT}$. As $\ringS$ and $\ringT$
        are Kroneckerian, the Galois closure of $\fieldfrac{\ringS}$
        as an extension of $\fieldfrac{\ringT}$ is necessarily
        Kroneckerian. As $\Lambda$ is normal, $\Lambda^\dagger$ has
        the same eigenvectors as $\Lambda$ with
        $\Lambda^\dagger\ket{\sigma(\alpha)} =
        \sigma(\alpha)^\dagger\ket{\sigma(\alpha)}=
        \sigma(\alpha^\dagger)\ket{\sigma(\alpha)}$ where we have used
        the defining property of Kroneckerian fields. For arbitrary
        $\ket{v}\in\hilbert_n$ we have
        \begin{align*}
            \Phi(M)^\dagger(\ket{v}\otimes\ket{\sigma(\alpha)}) = \sum_{j=0}^{d-1}M_j^\dagger\ket{v} \otimes (\Lambda^\dagger)^j\ket{\sigma(\alpha)}
            =\sum_{j=0}^{d-1}M_j^\dagger\ket{v} \otimes \sigma(\alpha^\dagger)^j\ket{\sigma(\alpha)}
            =\sigma(M^\dagger)\ket{v}\otimes\ket{\sigma(\alpha)}.
        \end{align*}
        On the other hand,
        \begin{align*}
            \Phi(M^\dagger)(\ket{v}\otimes\ket{\sigma(\alpha)}) = \sum_{j=0}^{d-1}N_j\ket{v} \otimes \Lambda^j\ket{\sigma(\alpha)}
            =\sum_{j=0}^{d-1}N_j\ket{v} \otimes \sigma(\alpha)^j\ket{\sigma(\alpha)}
            =\sigma(M^\dagger)\ket{v}\otimes\ket{\sigma(\alpha)}.
        \end{align*}
        We conclude $\Phi(M)^\dagger=\Phi(M^\dagger)$ since vectors of the form $\ket{v}\otimes\ket{\sigma(\alpha)}$ span $\hilbert_{kn}$.

      \item By definition of $\Phi$, $\Phi(C\otimes
        A)=C\otimes\Phi(A)$.
        
        \item We have $\alpha$ is an eigenvalue of $\Lambda$. Let
          $\Pi$ be the projector onto the eigenspace of $\Lambda$
          corresponding to $\alpha$. Then
        \[
            \Phi(M)(I\otimes\Pi) = \sum_{j=0}^{d-1} M_j\otimes \Lambda^j\Pi = \sum_{j=0}^{d-1} M_j\otimes \alpha^j\Pi = M\otimes\Pi.\qedhere
        \]
    \end{enumerate}
\end{proof}
 
\cref{thm:problem} provides necessary and sufficient conditions for
producing a linear catalytic embedding. More pertinently, it gives us
a road map for constructing linear catalytic embeddings. If we are
able to find a matrix $\Lambda$ with the properties listed above, we
can extend the map $\alpha\mapsto\Lambda$ to a linear catalytic
embedding on all of $\matrices(\ringS)$. Alternatively, if we can show
that no such $\Lambda$ exists, then this likewise precludes the
existence of a linear catalytic embeddings. While \cref{thm:problem}
shows how a linear catalytic embedding is determined by its action on
a single element, the following theorem gives a complete description
of the image of a catalytic embedding.

\begin{theorem}
    \label{thm:struc}
    Let $\ringR$ and $\ringS$ be Kroneckarian number rings, $\fieldK$
    be the Galois closure of
    $\fieldfrac{\ringS}/\fieldfrac{\ringS\cap\ringR}$,
    $\Phi:\matrices[\ringS]\rightarrow\matrices[\ringR]$ be the
    pre-embedding of a $k$-dimensional linear catalytic embedding, and
    $\{\tau_m\}_{m=1}^n$ be the set of automorphisms on
    $\fieldK\fieldfrac{\ringR}$ fixing $\fieldfrac{\ringR}$. There
    exists a set of automorphisms $\{\sigma_\ell\}_{\ell=1}^j$ on
    $\fieldK$ fixing $\fieldfrac{\ringS\cap\ringR}$ and corresponding
    projectors $\{\Pi_\ell\}_{\ell=1}^j\subset
    \matrices_k(\fieldK\fieldfrac{\ringR})$ with $1\leq j\leq
             [\fieldK\cap\fieldfrac{\ringR}:\fieldfrac{\ringS\cap\ringR}]$
             such that
    \begin{enumerate}
        \item $\{\tau_m(\Pi_\ell)~|~1\leq\ell\leq j,1\leq m\leq n\}$ is a complete set of mutually-orthogonal projectors,
        \item $\Phi(M)\left(I\otimes\tau_m(\projector_\ell)\right) = \left(\tau_m\circ\sigma_\ell(M)\right)\otimes \tau_m(\projector_\ell)$, and
        \item ${\Phi(M) = \sum_{\ell=1}^j \tr_{\fieldfrac{\sigma_{\ell}(\ringS)\ringR} / \fieldfrac{\ringR}}(\sigma_{\ell}(M)\otimes\Pi_\ell)}$.
    \end{enumerate}
    Moreover, the action of the circuit $\Phi(M)$ is completely
    determined by its action on the set $\{\Pi_\ell\}_{\ell=1}^j$.
\end{theorem}

\begin{proof}
    Let $\alpha\in\ringS$ such that
    $\fieldfrac{\ringS}=\fieldfrac{\ringS\cap\ringR}[\alpha]$ with
    minimal polynomial $p\in\fieldfrac{\ringS\cap\ringR}[x]$ with
    degree $d$. By the definition of Galois closure, $\fieldK$ is the
    splitting field of $p$. Suppose $p$ splits into irreducible
    factors $p=p_1p_2\dots p_r$ over $\fieldfrac{\ringR}$. We shall
    bound this $r$ value later. As $p$ cannot have repeated roots, no
    $p_\ell$ has repeated roots nor can two factors share any common
    roots. Let $\lambda_\ell$ be any single root of $p_\ell$ for
    $1\leq\ell\leq r$. Without loss of generality, we can assume
    $\lambda_1=\alpha$. For each $\lambda_\ell$, there exists an
    automorphism $\sigma_\ell: \fieldK\rightarrow\fieldK$ that fixes
    $\fieldfrac{\ringS}\cap\fieldfrac{\ringR}$ such that
    $\sigma_\ell(\alpha)=\lambda_\ell$. Without loss of generality, we
    take $\sigma_1$ to be identity.
    
    By \cref{thm:problem}, there exists a normal matrix
    $\Lambda\in\matrices(\ringR)$ such that
    \[
        p(\Lambda) = 0, \Phi(\alpha)=\Lambda,\text{ and }\Phi(M)=\sum_{i=0}^{d-1} M_i\otimes \Lambda^i \text{ where } M=\sum_{i=0}^{d-1} M_i\otimes \alpha^i\text{ and }M_i\in\matrices(\fieldfrac{\ringR}).
    \]
    Let $q\in\ringR[x]$ be the characteristic polynomial of $\Lambda$
    with degree $k$. Since $p(\Lambda)=0$, without loss of generality
    we can take $q=\pm p_1^{d_1}p_2^{d_2}\dots p_j^{d_j}$ with
    $d_\ell> 0$ and $j\leq r$. This is consistent with our choice to
    make $\alpha$ a root of $p_1$, as $\Phi(\alpha)=\Lambda$ and thus
    $\alpha$ is an eigenvalue of $\Lambda$. For $1\leq \ell\leq j$,
    observe that the eigenvector equation
    \[
        (\Lambda - \lambda_\ell I)v = 0
    \]
    has coefficients in the field $\fieldfrac{\ringR}[\lambda_\ell] =
    \fieldfrac{\sigma_\ell(\ringS)\ringR} \subseteq
    \fieldK\fieldfrac{\ringR}$. As
    $\fieldfrac{\sigma_\ell(\ringS)\ringR}$ is Kroneckerian, we can
    find a spanning set for all such $v$ over
    $\fieldfrac{\sigma_\ell(\ringS)\ringR}^k$ and apply the
    Gram-Schmidt procedure to obtain (unnormalized) eigenvectors
    $\{v^{(\ell)}_i\}_{i=1}^{d_\ell}$ with
    $v^{(\ell)}_i\in\fieldfrac{\sigma_\ell(\ringS)\ringR}^k$. This
    basis can be used to define the projector onto the eigenspace of
    $\lambda_\ell$,
    \[
        \Pi_{\ell}=\sum_{i=1}^{d_\ell}\frac{v^{(\ell)}_i\left(v^{(\ell)}_i\right)^\dagger}{\langle v^{(\ell)}_i,v^{(\ell)}_i\rangle}\in\matrices_k(\fieldfrac{\sigma_\ell(\ringS)\ringR})\subseteq\matrices_k(\fieldK\cap\fieldfrac{\ringR}),
    \]
    which is necessarily orthogonal. This immediately implies that
    application of any $\fieldfrac{\ringR}$-fixing
    $\fieldK\fieldfrac{\ringR}$-automorphism $\tau$ to any
    $\projector_\ell$ must also be an orthogonal projector, as
    \[
        \tau_m(\Pi_\ell)^\dagger\tau_m(\Pi_\ell) = \tau_m(\Pi_\ell^\dagger)\tau_m(\Pi_\ell) = \tau_m(\Pi_\ell^\dagger\Pi_\ell) = \tau_m(\Pi_\ell).
    \]
    
    As $\fieldK\fieldfrac{\ringR}/\fieldfrac{\ringR}$ is a Galois
    extension with intermediate field
    $\fieldfrac{\sigma_\ell(\ringS)\ringR}$,
    $\operatorname{Gal}(\fieldK\fieldfrac{\ringR}/\fieldfrac{\ringR})$
    has a subgroup $G_\ell$ which fixes
    $\fieldfrac{\sigma_\ell(\ringS)\ringR}$. By definition, $G_\ell$
    necessarily acts as identity on $\lambda_\ell$. Let
    $\tau_a,\tau_b\in\operatorname{Gal}(\fieldK\fieldfrac{\ringR}/\fieldfrac{\ringR})$. We
    necessarily have
    \[
        \tau_a G_\ell = \tau_b G_\ell\iff\tau_b^{-1}\tau_a\in G_\ell\iff (\tau_b^{-1}\tau_a)(\lambda_\ell)=\lambda_\ell \iff \tau_a(\lambda_\ell)=\tau_b(\lambda_\ell).
    \]
    By the fundamental theorem of Galois theory, we have
    \[
        |\operatorname{Gal}(\fieldK\fieldfrac{\ringR}/\fieldfrac{\ringR}):G_\ell| = [\fieldfrac{\sigma_\ell(\ringS)\ringR}:\fieldfrac{\ringR}]
    \]
    distinct cosets of $G_\ell$ in
    $\operatorname{Gal}(\fieldK\fieldfrac{\ringR}/\fieldfrac{\ringR})$,
    implying there are precisely
    $[\fieldfrac{\sigma_\ell(\ringS)\ringR}:\fieldfrac{\ringR}]$
    elements of $\{\tau_m(\lambda_\ell)\}_{m=1}^n$, each corresponding
    to exactly one root of $p_\ell$. By application of the
    automorphism $\tau_m$ to the eigenprojector equation for
    $\Lambda$, we find
    \[
        \Lambda \tau_m(\projector_\ell) = \tau_m(\Lambda \projector_\ell) = \tau_m(\lambda_\ell \projector_\ell) = \tau_m(\lambda_\ell) \tau_m(\projector_\ell)
    \]
    so that not only is $\tau_m(\projector_\ell)$ an orthogonal
    projector, it is an eigenprojector of $\Lambda$ with eigenvalue
    $\tau_m(\lambda_\ell)$. As $\tau_m$ is invertible,
    $\tau_m(\projector_\ell)$ must project onto the full eigenspace of
    $\tau_m(\sigma_\ell)$. Note that again there are only
    $[\fieldfrac{\sigma_\ell(\ringS)\ringR}:\fieldfrac{\ringR}]$
    distinct elements for the set $\{\tau_m(\Pi_\ell)~|~1\leq m\leq
    n\}$, one for each root of $p_\ell$.
    
    As $\Lambda$ is normal, its eigenspaces are mutually orthogonal,
    and hence we conclude that $\{\tau_m(\Pi_\ell)~|~1\leq\ell\leq
    j,1\leq m\leq n\}$ is a set of pairwise orthogonal
    projectors. Moreover, as there is a maximal projector for each
    root of $q$ and hence eigenvalue of $\Lambda$, this set must also
    be complete, which proves the first part of the theorem. For the
    second statement, direct computation yields
    \begin{align*}
        \Phi(M)\left(I\otimes\tau_m(\Pi_\ell)\right) &= \left(\sum_{i=0}^{d-1} M_i\otimes \Lambda^i\right)\left(I\otimes\tau_m(\Pi_\ell)\right) = \sum_{i=0}^{d-1} M_i\otimes (\Lambda^i\tau_m(\Pi_\ell))=\sum_{i=0}^{d-1} M_i\otimes \tau_m(\sigma_\ell(\alpha))^i\tau_m(\Pi_\ell)\\
        & =\sum_{i=0}^{d-1} M_i\tau_m(\sigma_\ell(\alpha))^i\otimes\tau_m(\Pi_\ell) =\sum_{i=0}^{d-1} \tau_m(\sigma_\ell(M_i\alpha^i))\otimes\tau_m(\Pi_\ell)=\tau_m\circ\sigma_\ell\left(\sum_{i=0}^{d-1} M_i\alpha^i\right)\otimes\tau_m(\Pi_\ell) \\
        &= \tau_m\circ\sigma_\ell(M)\otimes\tau_m(\Pi_\ell).
    \end{align*}
    Choosing one coset representative $\overline{\tau}_m^{(\ell)}$ for
    each coset of
    $\operatorname{Gal}(\fieldK\fieldfrac{\ringR}/\fieldfrac{\ringR})/G_\ell$,
    completeness of our projectors implies
    \[
        I = \sum_{\ell=1}^{j}\sum_{m=1}^{\deg(p_\ell)}\overline{\tau}_m^{(\ell)}(\Pi_\ell).
    \]
    Using this relation alongside the second statement, we conclude
    \begin{align*}
        \Phi(M) &= \Phi(M)(I\otimes I) = \Phi(M)\left(I\otimes \sum_{\ell=1}^{j}\sum_{m=1}^{\deg(p_\ell)}\overline{\tau}_m^{(\ell)}(\Pi_\ell)\right)\\
        &= \Phi(M)\left( \sum_{\ell=1}^{j}\sum_{m=1}^{\deg(p_\ell)}I\otimes\overline{\tau}_m^{(\ell)}(\Pi_\ell)\right) = \sum_{\ell=1}^{j}\sum_{m=1}^{\deg(p_\ell)}\Phi(M)\left(I\otimes\overline{\tau}_m^{(\ell)}(\Pi_\ell)\right)\\
        &=\sum_{\ell=1}^{j}\sum_{m=1}^{\deg(p_\ell)}\overline{\tau}_m^{(\ell)}(\sigma_\ell(M))\otimes\overline{\tau}_m^{(\ell)}(\Pi_\ell) = \sum_{\ell=1}^{j}\tr_{\fieldfrac{\sigma_{\ell}(\ringS)\ringR} / \fieldfrac{\ringR}}(\sigma_\ell(M)\otimes\Pi_\ell)
    \end{align*}
    by definition of the field trace, and hence we have proved the
    third major statement. Note that $\Phi(M)$ is completely
    decomposed into its action on orthogonal subspaces of the catalyst
    space by this statement, each determined by an automorphism of
    $\fieldK\fieldfrac{\ringR}$ and one projector from
    $\{\projector_\ell\}_{\ell=1}^j$. Again, completeness allows us to
    determine the action of $\Phi(M)$ for any input state for the
    catalyst space, proving the fourth statement.
    
    Finally, we bound $j\leq r$, the number of irreducible factors of
    $p$ over $\fieldfrac{\ringR}$. As $\fieldK$ is Galois over
    $\fieldfrac{\ringS\cap\ringR}$, there is a subgroup $H\leq
    \operatorname{Gal}(\fieldK/\fieldfrac{\ringS\cap\ringR})$ that
    fixes $\fieldK\cap\fieldfrac{\ringR}$. Because $H$ fixes
    $\fieldK\cap\fieldfrac{\ringR}$, elements of $H$ can \emph{only}
    permute the roots of irreducible polynomials over
    $\fieldK\cap\fieldfrac{\ringR}$. Therefore, because
    $\sigma_{\ell_1}(\alpha)$ and $\sigma_{\ell_2}(\alpha)$ are roots
    of different polynomials when ${\ell_1}\neq {\ell_2}$ by
    assumption, $\sigma_{\ell_1}H\neq \sigma_{\ell_2}H$ and each
    $\sigma_\ell H$ defines a different coset of $H$. By the
    fundamental theorem of Galois theory, there are precisely
    $[\fieldK\cap\fieldfrac{\ringR}:\fieldfrac{\ringS\cap\ringR}]$
    cosets, and hence $j\leq r =
    [\fieldK\cap\fieldfrac{\ringR}:\fieldfrac{\ringS\cap\ringR}]$.
\end{proof}

Breaking down \cref{thm:struc}, we see a complete characterization of
the structure of a linear catalytic embedding. Given a linear
catalytic embedding, statement one says that the catalyst subsystem
naturally decomposes into an orthonomal basis of ``catalyst-like''
states. The second statement shows that embedded circuits fix this
basis of the catalyst subsystem. In other words, the embedded circuit
is a controlled operator in the ``catalyst basis.'' But how does the
embedded circuit behave in the presence of the other basis states? The
second, third, and fourth statements tell us that an embedded circuit
acts like the original circuit composed with an automorphism of a
Galois field. These automorphisms map $\ringS$ to an isomorphic copy
of $\ringS,$ thereby changing the original unitary to some isomorphic
copy. These automorphisms also relate many of the projectors, so that
at most $[\fieldK\cap\fieldfrac{\ringR}:\fieldfrac{\ringS\cap\ringR}]$
need to be computed separately. This bound is a measure of how
polynomials factor over $\fieldfrac{\ringR}$ as opposed to
$\fieldfrac{\ringS\cap \ringR}$. If certain irreducible polynomials
over $\fieldfrac{\ringS\cap \ringR}$ factor over $\fieldfrac{\ringR}$,
more states will need to be computed. We can characterize precisely
when the construction simplifies, as the following corollary shows:

\begin{corollary}
    \label{cor:RSinterplay}
    Let $\ringR$ and $\ringS$ be Kroneckarian number rings, $\fieldK$
    be the Galois closure of
    $\fieldfrac{\ringS}/\fieldfrac{\ringS\cap\ringR}$,
    $\Phi:\matrices[\ringS]\rightarrow\matrices[\ringR]$ be the
    pre-embedding of a linear catalytic embedding, and
    $\{\tau_m\}_{m=1}^n$ be the set of automorphisms on
    $\fieldK\fieldfrac{\ringR}$ fixing $\fieldfrac{\ringR}$. If
    $\fieldfrac{\ringS\cap \ringR}=\fieldK\cap\fieldfrac{\ringR}$,
    then there exists a single orthogonal projector
    $\projector\in\matrices[\fieldfrac{\ringS\ringR}]$ such that
    \begin{enumerate}
    \item $\{\tau_m(\Pi)~|~1\leq m\leq n\}$ is a complete set of
      mutually-orthogonal projectors with
      $[\fieldfrac{\ringS}:\fieldfrac{\ringS\cap\ringR}]$ distinct
      elements,
    \item $\Phi(M)\left(I\otimes\tau_m(\projector)\right) =
      \tau_m(M)\otimes \tau_m(\projector)$ and,
    \item $\Phi(M)
      =\tr_{\fieldfrac{\ringS\ringR}/\fieldfrac{\ringR}}(M\otimes
      \projector)$.\label{eq:fieldtrace}
    \end{enumerate}
    Moreover, the action of the circuit $\Phi(M)$ is completely
    determined by its action on $\Pi$.
\end{corollary}

\begin{proof}
    As $\fieldK$ is a Galois extension of
    $\fieldfrac{\ringS\cap\ringR}=\fieldK\cap\fieldfrac{\ringR}$,
    $\fieldK$ and $\fieldfrac{\ringR}$ are linearly disjoint as
    extensions of their intersection
    $\fieldK\cap\fieldfrac{\ringR}$. Any subfield of $\fieldK$ is
    necessarily also linearly disjoint with $\fieldfrac{\ringR}$ as an
    extension of $\fieldK\cap\fieldfrac{\ringR}$, and thus
    $[\fieldfrac{\ringS\ringR}:\fieldfrac{\ringR}] =
    [\fieldfrac{\ringS}:\fieldfrac{\ringS\cap\ringR}]$. Applying
    \cref{thm:struc} with $\fieldfrac{\ringS\cap
      \ringR}=\fieldK\cap\fieldfrac{\ringR}$, we have $1\leq j\leq
         [\fieldK\cap\fieldfrac{\ringR}:\fieldfrac{\ringS\cap\ringR}]
         =1$, and so $j=1$ such that the only automorphism of
         $\fieldK\cap\fieldfrac{\ringR}$ fixing
         $\fieldfrac{\ringS\cap\ringR}$ is identity. We can thus write
         $\sigma_1(M) = M$ and $\projector_1 = \projector$, and the
         corollary follows by noting the number of distinct projectors
         is precisely $[\fieldfrac{\ringS\ringR}:\fieldfrac{\ringR}]$
         given the proof of \cref{thm:struc}.
\end{proof}

The hypothesis $\fieldfrac{\ringS\cap
  \ringR}=\fieldK\cap\fieldfrac{\ringR}$ is satisfied in most cases of
practical interest. For example, this condition holds trivially when
$\ringR\subset\ringS$, which coincides with the notion that our goal
with linear catalytic embeddings is to ``simplify'' the ring over
which we apply circuits. When $\fieldfrac{\ringS\cap
  \ringR}\neq\fieldK\cap\fieldfrac{\ringR}$, we are not so much
simplifying circuits as we are working (at least partially) over an
isomorphic copy of $\ringS$. For the remainder of this section, we
thus focus on illuminating the consequences of
\cref{thm:problem,cor:RSinterplay} with some examples to show linear
catalytic embeddings at work.

\begin{example}
\label{ex:linembed5}
    Let $\ringS = \Q[\sqrt{5}]$ and $\ringR=\Q$ so that
    $\ringR,\ringS$ are both their own field of fractions. The element
    $\sqrt{5}$ plays the role of $\alpha$ in \cref{thm:problem} with
    minimal polynomial $p=x^2-5$. We need to find a normal matrix
    $\Lambda$ such that $\Lambda^2 -5I=0$. The matrix
    \[\Lambda = \begin{bmatrix}
        1 & 2\\
        2 & -1
    \end{bmatrix}\]
    has characteristic polynomial $p$, and so $\Lambda^2-5I =
    0$. Additionally, $\Lambda$ has $\sqrt{5}$ as an eigenvalue and
    the set $\{1,\sqrt{5}\}$ is a generating set satisfying condition
    4 of \cref{thm:problem}. The pre-embedding
    $\Phi:\matrices(\Q[\sqrt{5}])\rightarrow\matrices(\Q)$ is
    constructed by extending the map $\sqrt{5}\mapsto \Lambda$ as:\[
    \Phi(M) = M_0\otimes\begin{bmatrix} 1 & 0 \\ 0 & 1
    \end{bmatrix} + M_1\otimes\begin{bmatrix}
        1 & 2\\
        2 & -1
    \end{bmatrix},
    \]
    where $M = M_0 + M_1\sqrt{5}$ with $M_0,M_1\in\matrices(\Q)$. As
    every quadratic number field extension is Galois, we have $\fieldK
    = \ringS = \Q[\sqrt{5}]$, whose automorphisms as an extension of
    $\Q$ are given by $\{\operatorname{id},\tau\}$ for
    $\tau:\sqrt{5}\mapsto-\sqrt{5}$ since $-\sqrt{5}$ is the only
    remaining root of $p$. Additionally, $\ringR\subset\ringS$ implies
    we can apply \cref{cor:RSinterplay} to determine the projectors
    and action of this embedding. We can easily compute the projector
    onto the eigenspace of $\Lambda$ corresponding to $\sqrt{5}$:
    \[ \Pi= 
    \begin{bmatrix}
        \frac{1}{2}\left(1+\frac{\sqrt{5}}{5}\right) & \frac{\sqrt{5}}{5}\\
        \frac{\sqrt{5}}{5} & \frac{1}{2}\left(1-\frac{\sqrt{5}}{5}\right)
    \end{bmatrix}.\]
    Checking the action of our embedding, we thus have
    \[
        \Phi(M)(I\otimes\Pi) = M_0\otimes (I\Pi) + M_1\otimes (\Lambda \Pi) = M_0\otimes \Pi + M_1\otimes (\sqrt{5} \Pi) = M\otimes \Pi
    \]
    as expected. Per \cref{cor:RSinterplay}, we should have that
    \[
    \Pi = \begin{bmatrix}
        \frac{1}{2}\left(1+\frac{\sqrt{5}}{5}\right) & \frac{\sqrt{5}}{5}\\
        \frac{\sqrt{5}}{5} & \frac{1}{2}\left(1-\frac{\sqrt{5}}{5}\right)
    \end{bmatrix}~~\xrightarrow{\tau}~~\begin{bmatrix}
        \frac{1}{2}\left(1-\frac{\sqrt{5}}{5}\right) & -\frac{\sqrt{5}}{5}\\
        -\frac{\sqrt{5}}{5} & \frac{1}{2}\left(1+\frac{\sqrt{5}}{5}\right)
    \end{bmatrix}=\tau(\Pi)
    \]
    is itself a projector so that $\{\projector,\tau(\projector)\}$ is
    a mutually orthogonal and complete set of projectors with
    cardinality $[\Q[\sqrt{5}]:\Q] = 2$. It is indeed clear that
    $\tau(\projector)\neq\projector$, and furthermore it is
    straightforward to verify that $\Pi^\dagger\Pi = \Pi$,
    $\tau(\Pi)^\dagger\tau(\Pi) = \tau(\Pi)$, $\tau(\Pi)^\dagger\Pi =
    0$, and $\projector+\tau(\projector) = I$ as expected.
    
    We can also check that our embedding suffices for extraction of
    the matrix $\tau(M) = M_0 - \sqrt{5} M_1$ via
    $\tau(\projector)$. Indeed, we have that $\Lambda \tau(\Pi) =
    -\sqrt{5}\tau(\Pi)$ so that
    \[
        \Phi(M)(I\otimes\tau(\Pi)) = M_0\otimes (I\tau(\Pi)) + M_1\otimes (\Lambda \tau(\Pi)) = M_0\otimes \tau(\Pi) + M_1\otimes (-\sqrt{5} \tau(\Pi)) = \tau(M)\otimes \tau(\Pi).
    \]
    Finally, by completeness of our projectors $\Pi$ and $\tau(\Pi)$,
    we have
    \[
        \Phi(M) = \Phi(M)(I\otimes(\Pi + \tau(\Pi))) = M\otimes \Pi + \tau(M)\otimes\tau(\Pi) = \tr_{\Q[\sqrt{5}]/\Q}(M\otimes \Pi)
    \]
    which confirms that the image of $M$ under $\Phi$ is given by the
    expected field trace.
\end{example}

\begin{example}
\label{ex:linembedomega}
    Let $\omega = e^{2\pi i/8}$ and $\ringR = \D = \Z[1/2]$, $\ringS =
    \D[i = \omega^2]$, and $\ringT = \D[\omega]$. We will construct
    the pre-embeddings of linear catalytic embeddings
    $\Phi_1:\matrices[\ringT]\rightarrow\matrices[\ringS]$ and
    $\Phi_2:\matrices[\ringS]\rightarrow\matrices[\ringR]$ to
    construct the pre-embedding
    $\Phi_2\circ\Phi_1:\matrices[\ringT]\rightarrow\matrices[\ringR]$. In
    light of number-theoretic characterizations of certain
    fault-tolerant gate sets \cite{Giles2013a,Amy2020}, this
    construction yields a structure-preserving map from the
    Clifford-$T$ gate set ($\D[\omega]$) to subsets $\D[i]$ and $\D$
    of the Clifford+$CS$ and Hadamard+Toffoli gate sets, respectively.
    
    We begin by finding the pre-embedding
    $\Phi_1:\matrices(\ringT)\rightarrow\matrices(\ringS)$. In this
    case, $\fieldfrac{\D[\omega]} = \Q[\omega]$ and
    $\fieldfrac{\D[i]}= \Q[i]$. The element $\omega$ plays the role of
    $\alpha$ from \cref{thm:problem} with minimal polynomial
    $p_1=x^2+i$, and the set $\{1,\omega\}$ plays the role of
    $\Gamma$. The normal matrix
    \[
        \Lambda_1 = \begin{bmatrix}
        0 & 1\\
        i & 0
    \end{bmatrix}
    \]
    has $p_1$ as its characteristic polynomial, and so satisfies $p_1$
    as required. In addition $\omega$ is an eigenvalue of $\Lambda_1$
    and so $\Lambda_1$ satisfies all the properties of
    \cref{thm:problem}. The linear catalytic embedding
    $\Phi_1:\matrices(\ringT)\rightarrow\matrices(\ringS)$ is
    constructed by extending the map $\omega\mapsto\Lambda_1$ as
    follows:
    \[
    \Phi_1(M) = M_0\otimes I + M_1\otimes\Lambda_1\text{ where } M = M_0 + M_1\omega\text{ with } M_0,M_1\in\matrices(\ringS).
    \]
    
    We have that the extension $\Q[\omega]/\Q[i]$ is Galois so that
    there are two distinct automorphisms of $\Q[\omega]$ which fix
    $\Q[i]$. These are given by
    \[
        \left\{\operatorname{id}:\omega\mapsto\omega, \tau_1:\omega\mapsto-\omega, \right\}.
    \]
    As we have $\ringS\subset\ringT$, we can apply
    \cref{cor:RSinterplay}. The projector $\Pi_1$ is the projector
    onto the subspace of $\Lambda_1$ corresponding to $\omega$ given
    by
    \[
    \Pi_1 = \frac{1}{2}
    \begin{bmatrix}
        1&\omega^7\\
        \omega & 1 
    \end{bmatrix}.\]
    It is a straightforward exercise to check that $\tau_1(\Pi_1)$ is
    the projector onto the eigenspace of $\Lambda_1$ corresponding to
    $\tau_1(\omega) = -\omega$, and furthermore that these projectors
    are mutually orthogonal, complete, and satisfy their respective
    catalytic conditions.
    
    Constructing the pre-embedding
    $\Phi_2:\matrices(\D[i])\rightarrow\matrices(\D)$ is remarkably
    similar. As $\fieldfrac{\D} = \Q$, the element $i$ plays the role
    of $\alpha$ from \cref{thm:problem} with minimal polynomial
    $p_2=x^2+1$, and the set $\{1,i\}$ plays the role of $\Gamma$. The
    normal matrix
    \[
        \Lambda_2 = \begin{bmatrix}
        0 & 1\\
        -1 & 0
    \end{bmatrix}
    \]
    satisfies the conditions of \cref{thm:problem}, and so we can
    define $\Phi_2:\matrices(\ringS)\rightarrow\matrices(\ringR)$ by
    \[
        \Phi_2(M) = M_0\otimes I + M_1\otimes\Lambda_2\text{ where } M = M_0 + M_1 i\text{ with } M_0,M_1\in\matrices(\ringR).
    \]
    
    Again, the extension $\Q[i]/\Q$ is Galois so that there are two
    distinct automorphisms of $\Q[i]$ which fix $\Q$. These are
    \[
        \left\{\operatorname{id}:i\mapsto i, \tau_2: i\mapsto-i, \right\}.
    \]
    Since $\ringR\subset\ringS$, we can again apply
    \cref{cor:RSinterplay}. The projector $\Pi_2$ projects onto the
    subspace of $\Lambda_2$ corresponding to $i$ given by
    \[
    \Pi_2 = \frac{1}{2}
    \begin{bmatrix}
        1& i^3\\
        i & 1 
    \end{bmatrix}.\]
    As before, $\tau_2(\Pi_2)$ projects onto the eigenspace of
    $\Lambda_2$ corresponding to $\tau_2(i) = -i$, and the projectors
    are mutually orthogonal, complete, and satisfy their respective
    catalytic conditions.
    
    We now check that the concatenation $\Phi_2\circ\Phi_2$ behaves as
    expected and yields a pre-embedding
    $\Phi:\matrices[\ringT]\rightarrow\matrices[\ringR]$ for a linear
    catalytic embedding. Given
    \[
        M = M_0 + \omega M_1 + \omega^2 M_2 + \omega^3 M_3
    \]
    for the basis $\{1,\omega,\omega^2,\omega^3\}$ of $\ringT$ as an
    $\ringR$ module we have
    \begin{align*}
        \Phi(M) &= \Phi_2\circ\Phi_1(M)\\
        &= \Phi_2((M_0+ \omega^2 M_2)\otimes I_2 + (M_1  + \omega^2 M_3)\otimes \Lambda_1)\\
        &= M_0\otimes I_4 + M_1\otimes\Phi_2(\Lambda_1) + M_2\otimes I_2\otimes\Lambda_2 + M_3\otimes\Phi_2(i\Lambda_1).
    \end{align*}
    Computing $\Phi_2(\Lambda_1)$, we have
    \[
        \Phi_2(\Lambda_1) = \Lambda = \begin{bmatrix}
            0 & 0 & 1 & 0\\
            0 & 0 & 0 & 1\\
            0 & 1 & 0 & 0\\
            -1 & 0 & 0 & 0
        \end{bmatrix}.
    \]
    In fact, we have $\Lambda^2 = I_2\otimes\Lambda_2$ and $\Lambda^3
    = \Phi_2(i\Lambda_1)$ so that really
    \[
        \Phi(M) =  M_0\otimes \Lambda + M_1\otimes\Lambda + M_2\otimes \Lambda^2 + M_3\otimes\Lambda^3.
    \]
    We thus identify $\Lambda$ as the normal matrix $\Phi(\omega)$,
    and indeed the characteristic polynomial of $\Lambda$ is $x^4+1$,
    which is precisely the minimal polynomial of $\omega$ over
    $\Q$. Furthermore, the projector
    \[
        \Pi_1\otimes\Pi_2 = \projector = \frac{1}{4}\begin{bmatrix}
            1 & \omega^6 & \omega^7 & \omega^5\\
            \omega^2 & 1 & \omega & \omega^7\\
            \omega & \omega^7 & 1 & \omega^6\\
            \omega^3 & \omega & \omega^2 & 1
        \end{bmatrix}
    \]
    projects onto the $\omega$ eigenspace of $\Lambda$ as
    expected. The automorphisms of $\Q[\omega]$ which fix $\Q$ are
    \[
        \left\{\operatorname{id}:\omega\mapsto\omega, \sigma_1:\omega\mapsto\omega^3, \sigma_2:\omega\mapsto\omega^5, \sigma_3:\omega\mapsto \omega^7   \right\}.
    \]
    and their action on $\Pi$ yield orthogonal projectors which are
    mutually orthogonal, complete, and satisfy the appropriate
    catalytic conditions. Thus we have constructed the desired
    pre-embedding via concatenation.
\end{example}

\section{Standard Catalytic Embeddings}
\label{sec:standard}

Ring extensions are often built by taking quotients. This is
especially true in the context of the number rings encountered in
fault-tolerant quantum computing \cite{Amy2020}. In these cases, there
is a convenient method to build linear catalytic embeddings, which can
be seen as a simplification of \cref{thm:problem}. We call the linear
catalytic embeddings defined in this way \emph{standard catalytic
embeddings}.

\subsection{Definitions and Properties}
\label{ssec:standarddef}

We start by introducing a generalization of the notion of companion
matrix which will play, in the context of standard catalytic
embeddings, the role of matrix $\Lambda$ in \cref{thm:problem}, as
seen in \cref{ex:linembed5,ex:linembedomega}.

\begin{definition}[Pseudo-Companion Matrix]
  \label{def:pseudocomp}
  Let $\ringR$ be a number ring and let $p\in\ringR[X]$ be a monic
  polynomial over $\ringR$. A matrix $\Lambda\in\matrices[\ringR]$ is
  a \emph{pseudo-companion matrix for $p$} if the characteristic
  polynomial of $\Lambda$ is $\pm p^c$, for some positive integer $c$.
\end{definition}

\begin{proposition}
    \label{prop:lamblin}
    Let $\ringR\subset\ringR[\alpha]$ be a Kroneckerian, integral
    extension of number rings such that $\alpha$ has minimal
    polynomial $p\in\fieldfrac{\ringR}[x]$ with coefficients in
    $\ringR$ and degree $d$. Suppose $\Lambda\in\matrices(\ringR)$ is
    a normal pseudo-companion matrix for $p$. If
    $A\in\matrices(\ringR[\alpha])$, then:
    \begin{enumerate}
        \item there exist unique $A_i\in\matrices(\ringR)$ such
          that \[ A =\sum_{i=0}^{d-1} A_i\alpha^i,
        \]
        \item the map
          $\Phi:\matrices(\ringR[\alpha])\rightarrow\matrices(\ringR)$
          given by \[ A\mapsto \sum_{i=0}^k A_i\otimes\Lambda^i
        \] is the pre-embedding for a linear catalytic embedding as defined in \cref{prop:embedformal}.
    \end{enumerate}
\end{proposition}

This result follows from \cref{thm:problem}, but we outline the proof
here.

\begin{proof}
    Let $A,B\in\matrices(\ringR[\alpha])$ and
    $C\in\matrices(\ringR)$. Because $p$ is irreducible over
    $\fieldfrac{\ringR}$ and has coefficients in $\ringR$,
    $\ringR[\alpha]$ is a free module over $\ringR$ with basis
    $\{\alpha^i\}_{i=0}^{d-1}$. Therefore, for each
    $A\in\matrices(\ringR[\alpha])$, there exist unique
    $A_i\in\matrices(\ringR)$ such that \[ A =\sum_{i=0}^{d-1}
    A_i\alpha^i.
        \]
    We now show that $\Phi$ satisfies the conditions of a
    pre-embedding of a linear catalytic embedding.

    \begin{enumerate}
        \item Because $\Lambda$ is normal and has characteristic
          polynomial $\pm p^m$ for some positive integer $m$, the
          eigenvalues of $\Lambda$ are all the roots of $p$, each
          occurring with multiplicity $m$. In addition, the
          eigenvectors of $\Lambda$ span $\hilbert_k$. Let $\ket{x}$
          be an eigenvector of $\Lambda$ with eigenvalue $\lambda$ and
          $\ket{v}\in\hilbert_n$ be an arbitrary state. Because $p$ is
          irreducible, there always exists a ring isomorphism
          $\sigma:\ringR[\alpha]\rightarrow\ringR[\lambda]$ that fixes
          $\ringR$ and maps $\sigma(\alpha)=\lambda$. Then
        \[
        \Phi(M)(\ket{v}\otimes\ket{x}) = \sum_{i=0}^{d-1} M_i\ket{v}\otimes \Lambda^i\ket{x} = \sum_{i=0}^{d-1} M_i\ket{v}\otimes \lambda^i\ket{x} = \sum_{i=0}^{d-1} M_i\ket{v}\otimes \sigma(\alpha)^i\ket{x} = (\sigma(M)\ket{v})\otimes \ket{x}.
        \]
        Because $\sigma$ is a ring isomorphism and vectors of the form
        $\ket{v}\otimes\ket{x}$ span $\hilbert_{kn}$, we conclude that
        $\Phi(AB)=\Phi(A)\Phi(B)$ when $AB$ is defined.
        
        \item $\displaystyle{\Phi(A)+\Phi(B)=\sum_{i=0}^{d-1} A_i\otimes\Lambda^i+\sum_{i=0}^{d-1} B_i\otimes\Lambda^i=\sum_{i=0}^{d-1} (A_i+B_i)\otimes\Lambda^i=\Phi(A+B)}$.

        \item $\displaystyle{\Phi(I_n) = \Phi(I_n\alpha^0)=I_n\otimes\Lambda^0=I_{nk}}.$

        \item Because $\Lambda$ is normal, if $\lambda$ is an
          eigenvalue of $\Lambda$ with eigenvector $\ket{x}$, then
          $\lambda^\dagger$ is an eigenvalue of $\Lambda^\dagger$ with
          eigenvector $\ket{x}$. We have
        \[
        M^\dagger = \sum_{i=0}^{d-1} M_i^\dagger (\alpha^i)^\dagger\quad\text{and}\quad M^\dagger = \sum_{i=0}^{d-1} N_i\alpha^i
        \]
        for some unique $N_i\in\ringR$. If $\ket{v}\in\hilbert_n$ is
        some arbitrary state and $\sigma$ is defined as above,
        \[
        \Phi(M)^\dagger(\ket{v}\otimes \ket{x}) = \sum_{i=0}^{d-1} M_i^\dagger \ket{v} \otimes (\Lambda^i)^\dagger \ket{x} = \sum_{i=0}^{d-1} M_i^\dagger \ket{v} \otimes (\lambda^i)^\dagger \ket{x} = \sum_{i=0}^{d-1} M_i^\dagger \ket{v} \otimes \sigma(\alpha^i)^\dagger \ket{x} = (\sigma(M^\dagger)\ket{v})\otimes\ket{x}.
        \]
        On the other hand,
        \[
            \Phi(M^\dagger)(\ket{v}\otimes\ket{x}) = \sum_{i=0}^{d-1}N_i\ket{v} \otimes \Lambda^i\ket{x} =\sum_{i=0}^{d-1}N_i\ket{v} \otimes \lambda^i\ket{x}
            =\sum_{i=0}^{d-1}N_i\ket{v} \otimes \sigma(\alpha)^i\ket{x}
            =\sigma(M^\dagger)\ket{v}\otimes\ket{x}.
        \]
        We conclude that $\Phi(M)^\dagger=\Phi(M^\dagger)$, since
        vectors of the form $\ket{v}\otimes\ket{x}$ span
        $\hilbert_{kn}$.

        \item $\displaystyle{\Phi(C\otimes A) = \Phi\left(\sum_{i=0}^{d-1} C\otimes A_i\alpha^i\right)=\sum_{i=0}^{d-1} C\otimes A_i\otimes\Lambda^i=C\otimes\Phi(A)}$.

        \item Because $\Lambda$ is a pseudo-companion matrix for $p$,
          it has $\alpha$ as an eigenvalue. Let $\Pi$ be the projector
          of onto the eigenspace of $\Lambda$ corresponding
          $\alpha$. Then \[\displaystyle{\Phi(A)(I\otimes\Pi)=\left(\sum_{i=0}^{d-1}
            A_i\otimes\Lambda^i\right)(I\otimes\Pi)=\sum_{i=0}^{d-1}
            A_i\otimes(\alpha^i\Pi)=\sum_{i=0}^{d-1}
            A_i\alpha^i\otimes\Pi=\Phi(A)\otimes\Pi}.\qedhere\]
    \end{enumerate}
\end{proof}

\begin{corollary}
    \label{cor:normtounit}
    Let $\ringR$, $\alpha$, $\Lambda$, and $\Phi$ be defined as in
    \cref{prop:lamblin}. If $U\in\unitary(\ringR[\alpha])$, then
    $\Phi(U)\in\unitary(\ringR)$.
\end{corollary}

\begin{proof}
    By \cref{prop:lamblin}, $\Phi(U^\dagger)=\Phi(U)^\dagger$,
    $\Phi(UV)=\Phi(U)\Phi(V)$, and $\Phi(I)=I_n$ for some $n$. Suppose
    $U\in\unitary(\ringR[\alpha])$, then \[\Phi(U)\Phi(U)^\dagger =
    \Phi(U)\Phi(U^\dagger)=\Phi(UU^\dagger)=\Phi(I)=I_n,\] and
    $\Phi(U)$ is unitary.
\end{proof}

\begin{definition}[Standard catalytic embedding]
    Let $\ringR$, $\alpha$, $\Lambda$, and $\Phi$ be defined as in
    \cref{prop:lamblin}, and let $\phi$ be the catalytic embedding
    obtained from the lifting of
    $\Phi|_{\unitary(\ringR[\alpha])}$. We call this embedding a
    \emph{standard} catalytic embedding.
\end{definition}

The difficulty in constructing a standard catalytic embedding hinges
on finding a normal pseudo-companion matrix over a given number
ring. On the surface, this might not seem immediately difficult. Each
polynomial over a number ring automatically admits a companion matrix,
and hence a pseudo-companion matrix. The normality condition ensures
unitarity of the resulting standard catalytic embedding, and is where
any difficulty in constructing such an embedding lies. In fact,
generalizing to pseudo-companion matrices (rather than simply
companion matrices) is \emph{necessary} to accommodate the normality
condition, as the following remark shows.

\begin{remark}
\label{rmk:nocompanion}
An extension of $\Q$ by a totally real algebraic number $\theta$
yields a totally real number field, and this extension $\Q[\theta]$
necessarily satisfies the assumptions of
\cref{prop:standard}. However, there exist $\theta$ whose minimal
polynomials do not permit a normal companion matrix with entries in
$\Q$ \cite{Schmeisser1993}. An example of such a polynomial is
$X^2-3$, because 3 cannot be written as a sum of two squares from
$\Q$. In fact, this is a relatively universal phenomenon, as randomly
selecting an irreducible quadratic polynomial over $\Q$ will yield a
polynomial $p$ that does not permit a normal companion matrix with
high probability. Using pseudo-companion matrices, rather than the
usual companion matrices, skirts this issue.
\end{remark}

We also might wonder under what conditions a generic number ring
extension $\ringR\subset\ringS$ permits an $\alpha\in\ringS$ as in
\cref{prop:lamblin} so that $\ringS=\ringR[\alpha]$. As it turns out,
these conditions are rather strong. Nonetheless, constructing ring
extensions in this way coincides with many ring extensions of
practical use. For example, quadratic extensions and extensions by
roots of unity generically have the required structure.

\begin{example}
\label{ex:std}
We show how to implement a fifth root of unity $\alpha = e^{2\pi i/5}$
over the Clifford+$T$ gate set using a standard catalytic
embedding. The Clifford+$T$ gate set has a characterization given by
$\unitary(\ringR)$ where $\ringR = \Z[1/2,\sqrt{2},i]$
\cite{Giles2013a}. Therefore, we need to construct a matrix whose
characteristic polynomial is a power of the minimal polynomial of
$\alpha$ over $\ringR$. The minimal polynomial of $\alpha$ over
$\ringR$ is $p(x) = x^4+x^3+x^2+x+1$. The matrix
\[
\Lambda = \frac{1}{2}\begin{bmatrix}
-1+i & 1 & 0 & i\\
1 & i & i & i\\
0 & i & -1-i & 1\\
i & i & 1 & -i
\end{bmatrix}
\]
is normal, is a pseudo-companion matrix for $p$, and has entries in
$\ringR$. The mapping
\[
U = \sum_{k=0}^4 A_k e^{k\frac{2\pi i}{5}} \mapsto \sum_{k=0}^4
A_k\otimes \Lambda^k
\]
can then be lifted to a standard catalytic embedding of
$\circuits(\ringR[\alpha])$ into $\circuits(\ringR)$. The catalyst for
this embedding is the eigenvector of $\Lambda$ corresponding to the
eigenvalue $\alpha$.
\end{example}

\begin{proposition}
\label{prop:intext}
       Let $\ringR\subset\ringR[\alpha]\subset\ringR[\beta]$ be
       integral Kroneckerian extensions. If $\Phi_1:
       \matrices(\ringR[\beta]) \rightarrow \matrices(\ringR[\alpha])$
       and $\Phi_2: \matrices(\ringR[\alpha]) \rightarrow
       \matrices(\ringR)$ are pre-embeddings for standard catalytic
       embeddings, then $\Phi=\Phi_2\circ\Phi_1$ is a pre-embedding
       for a standard catalytic embedding.
\end{proposition}

\begin{proof}
    Let $p$ be the minimal polynomial for $\beta$ over $\ringR$. Let
    $\Lambda = \Phi_2\circ\Phi_1(\beta)\in \matrices(\ringR)$. Then $0
    = \Phi_2\circ\Phi_1(p(\beta)) = p(\Phi_2\circ\Phi_1(\beta)) =
    p(\Lambda)$. Because $p$ is irreducible over $\ringR$, the
    characteristic polynomial of $\Lambda$ must be a power of $p$, so
    $\Lambda$ is a pseudo-companion matrix for $p$.  Because $\Phi_2$
    and $\Phi_1$ are pre-embeddings, \[\Lambda\Lambda^\dagger =
    \Phi_2\circ\Phi_1(\beta)\Phi_2\circ\Phi_1(\beta^\dagger)=\Phi_2\circ\Phi_1(\beta\beta^\dagger)=\Phi_2\circ\Phi_1(\beta^\dagger\beta)=\Phi_2\circ\Phi_1(\beta^\dagger)\Phi_2\circ\Phi_1(\beta)=\Lambda^\dagger\Lambda,\]
    thus $\Lambda$ is normal.  Let
    $M\in\matrices(\ringR[\beta])$. There exist $M_i$ such that
    $\displaystyle{M = \sum_{i=1}^{d-1}M_i\beta^i}$ where $d$ is the
    degree of $p$. By linearity of $\Phi_1$ and $\Phi_2$,
    \[{\Phi_2\circ\Phi_1(M) = \sum_{i=1}^{d-1}M_i\otimes\Phi_2\circ\Phi_1(\beta)^i=\sum_{i=1}^{d-1}M_i\otimes\Lambda^i}.
    \]
    Thus $\Phi=\Phi_2\circ\Phi_1$ is a pre-embedding of a standard catalytic embedding.
\end{proof}

\begin{example}
In \cref{ex:std}, while $\Lambda$ is, in fact, a companion matrix for
$p$, it is not immediately obvious how to arrive at $\Lambda$ given
$p$. Here we show how to apply \cref{prop:intext} to construct
$\Lambda$ using an intermediate ring extension.

The ring $\ringR[\cos(2\pi/5)]$ sits between $\ringR[\alpha]$ and
$\ringR$. First, we find a matrix $\Lambda_1$ for the embedding from
$\unitary(\ringR[\alpha])$ to $\unitary(\ringR[\cos(2\pi/5])$. The
minimal polynomial of $\alpha$ over $\ringR[\cos(2\pi/5)]$ is
$q(x)=x^2+2\cos(2\pi/5)x+1$. Our matrix must therefore have $\alpha$
and $\overline{\alpha}$ as its only eigenvalues. The matrix \[
\frac{1}{2}\begin{bmatrix} 1 & 1 +2\cos(2\pi/5)\\ 1+2\cos(2\pi/5) & -1
\end{bmatrix}
\]
has eigenvalues $\pm\sin(2\pi/5)$ (since $\sin^2(2\pi/5)=(1/2)^2 +
(1/2+\cos(2\pi/5))^2$). Then, by multiplying by $i$ and adding
$\cos(2\pi/5)\cdot I$ to this matrix, we obtain the matrix
\[
\Lambda_1 = \frac{i}{2}\begin{bmatrix}
    1 & 1 +2\cos(2\pi/5)\\
    1+2\cos(2\pi/5) & -1
\end{bmatrix} + \begin{bmatrix}
    \cos(2\pi/5) & 0\\
    0 & \cos(2\pi/5)
\end{bmatrix} = \frac{1}{2}\begin{bmatrix}
    i + 2\cos(2\pi/5) & i + 2i\cos(2\pi/5)\\
    i + 2i\cos(2\pi/5) & -i + 2\cos(2\pi/5)
\end{bmatrix},
\]
whose eigenvalues are $\cos(2\pi/5) \pm
i\sin(2\pi/5)=\alpha,\overline{\alpha}$. Consequently, $\Lambda_1$ has
the characteristic polynomial, $q$.

Next, we find a matrix $\Lambda_2$ for the embedding from
$\unitary(\ringR[\cos(2\pi/5])$ to $\unitary(\ringR)$. The minimal
polynomial of $\cos(2\pi/5)$ over $\ringR$ is $r(x) = x^2 +
\frac{1}{2}x + \frac{1}{4}$. The matrix \[ \Lambda_2 =
\frac{1}{2}\begin{bmatrix} -1 & 1 \\ 1 & 0
\end{bmatrix}
\]
has $r$ as its characteristic polynomial.

Finally, we embed $\Lambda_1$ using $\Lambda_2$ in order to obtain
\[
\Lambda_1 = \frac{1}{2} \begin{bmatrix}
    i & i \\
    i & -i
\end{bmatrix} + \begin{bmatrix}
    1 & i\\
    i & 1
\end{bmatrix}\cdot \cos(2\pi/5) \longmapsto \frac{1}{2} \begin{bmatrix}
    i & i \\
    i & -i
\end{bmatrix}\otimes I_2 + \begin{bmatrix}
    1 & i\\
    i & 1
\end{bmatrix}\otimes \Lambda_2 = \frac{1}{2}\begin{bmatrix}
-1+i & 1 & 0 & i\\
1 & i & i & i\\
0 & i & -1-i & 1\\
i & i & 1 & -i
\end{bmatrix} = \Lambda.
\]
If $\ket{\psi_1}$ is the eigenvector of $\Lambda_1$ corresponding to
$\alpha$ and $\ket{psi_2}$ is the eigenvector of $\Lambda_2$
corresponding to $\cos(2\pi/5)$, then the catalyst for $\Lambda$ is
given by $\ket{\psi_1}\otimes\ket{\psi_2}$.
\end{example}

\begin{theorem}
    \label{prop:standard}
    Let $\ringR\subset\ringR[\alpha]$ be a Kroneckerian, integral
    extension of number rings such that $\alpha$ has a minimal
    polynomial over $\fieldfrac{\ringR}$ with coefficients in
    $\ringR$. If $\phi: \circuits(\ringR[\alpha]) \rightarrow
    \circuits(\ringR)$ is a linear catalytic embedding, then $\phi$ is
    a standard catalytic embedding.
\end{theorem}

\begin{proof}
    Let $\Phi$ be the pre-embedding of $\phi$. The set
    $\{\alpha^j\}_{j=0}^k$ is a generating set for $\ringR[\alpha]$
    over $\ringR$ for some $k$. Let $\Phi(\alpha)=\Lambda$ and
    $p\in\ringR$ be the minimal polynomial of $\alpha$. By
    \cref{thm:problem}, $\Lambda$ is normal and $p(\Lambda) =
    0$. Because the characteristic polynomial of $\Lambda$ has
    coefficients in $\fieldfrac{\ringR}$ and is divisible by $p$ with
    $p$ irreducible over $\fieldfrac{\ringR}$, the characteristic
    polynomial of $\Lambda$ is a plus or minus a power of
    $p$. Therefore, $\Lambda$ is a normal pseudo-companion matrix for
    $p$. Thus, by \cref{prop:onlyS} $\Phi$ is completely determined by
    its action on $\{\alpha^j\}_{j=0}^k$, and so $\Phi$ is the
    pre-embedding of a standard catalytic embedding.
\end{proof}

Standard catalytic embeddings are useful because they reduce the
problem of constructing linear catalytic embeddings to the problem of
finding normal pseudo-companion matrices for irreducible
polynomials. The latter problem is, in general, simpler. When the
rings in question are actually fields, there is an explicit method for
constructing such matrices, and thus an explicit methods for building
standard catalytic embeddings \cite{ringsandfields}.

\subsection{Circuits Including Order-3 \texorpdfstring{$Z$}{Z}-Rotations}
\label{sect:domegathree}

One can obtain explicit reductions in the resources needed for
fault-tolerant quantum computing using standard catalytic
embeddings. We consider computations over the Clifford+$T$ gates set,
a gate set commonly used in fault-tolerant quantum computing. Let
$\omega_3$ be a third root of unity and $E$ be the rotation
\[
E:= \begin{bmatrix} 1 & 0 \\ 0 & \omega_3
	\end{bmatrix}.
\] 
It was shown in \cite{Giles2013a} that the Clifford+$T$ gate set
corresponds to the set of unitary operations $\unitary(\D[\omega_8])$
where $\omega_8$ is an eighth root of unity and $\D=\Z[1/2]$. Since
$\omega_3\notin \D[\omega_8]$, it follows that $E$ cannot be
implemented directly by a circuit over Clifford+$T$.

While quantum computation over Clifford+$T$+$E$ may not be typical, in
some circumstances it may be useful to extend the Clifford+$T$ gate
set with a phase gate of order $3$.  In the standard approach, one
would implement the $E$ gate by approximation over Clifford+$T$.  With
embeddings, we can instead directly implement $E$ using Clifford+$T$
gates by embedding $\unitary(\D[\omega_8,\omega_3])$ in
$\unitary(\D[\omega_8])$.

Because $\omega_3\notin \D[\omega_8]$ and $\omega_3^2+\omega_3+1=0$,
we need to find a normal, pseudo-companion matrix for the polynomial
$x^2 +x+1$ with entries in the ring $\D[\omega_8]$. Explicitly, we can
observe that
\[
	\Lambda = \frac{1}{2}\begin{bmatrix}  -1 - i & 1 - i \\ -1 - i & -1 + i \end{bmatrix}
\]
is a normal pseudo-companion matrix for $\omega_3$ over $\D[\omega_8]$
with $i=\omega_8^2$, giving a standard catalytic 
embedding. Note also that
\[
	\ket{v} = \frac{1}{\sqrt{3 + \sqrt{3}}}\begin{bmatrix} -\omega_3 - i\omega_3^2 \\ 1 \end{bmatrix}
\]
is in the $\omega_3$ eigenspace of $\Lambda$ and the corresponding
projector is given by $\ket{v}\bra{v}$.  Finally, we note that
\[
	\Lambda = \omega^5_8HS.
\]
Since $E = \ket{0}\bra{0} + \omega_3 \ket{1}\bra{1}$, we have
\begin{align*}
	\phi(E) = \ket{0}\bra{0}\otimes I + \ket{1}\bra{1}\otimes\Lambda 
		= \ket{0}\bra{0}\otimes I + \ket{1}\bra{1}\otimes (\omega^5_8HS).
\end{align*}
In circuit form, the right hand side amounts to a controlled
$\omega^5_8HS$ gate, as synthesized below (a controlled $\omega_8^5$
gate can be implemented as a $T^5 = ZT$ gate):
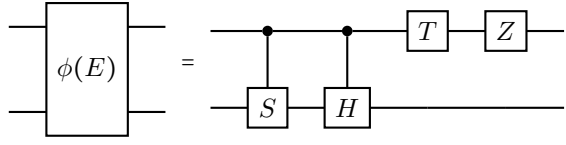
\[
\begin{quantikz}
        & \gate[2]{\phi(E)} & \qw \\
        & & \qw
\end{quantikz}
=
\begin{quantikz}
        & \ctrl{1} & \ctrl{1} &  \gate{T} & \gate{Z} & \qw \\
	& \gate{S} & \gate{H} & \qw & \qw & \qw
\end{quantikz}
\]
Using known constructions for the controlled-$H$ and -$S$ gates
\cite{ammr13} gives a Clifford+$T$ implementation of $\phi(E)$ using
$6$ $T$ gates, which can be further reduced to $4$ $T$ gates using
standard techniques (e.g. \cite{amm14}).

For circuits with a large number of $E$ gates, this gives a
significant reduction in $T$-count compared with the usual method of
repeated approximations. In particular, for a circuit with $m$ $E$
gates, assuming a practical overall precision of
$\varepsilon=10^{-15}$ \cite{kmm13}, and using asymptotically optimal
Clifford+$T$ approximations \cite{ross2016optimal} of $E$, this gives
a $T$-count of
\[
m\cdot 3\log_2(m/\varepsilon) \approx 3m\cdot(50+\log_2 m).
\] 
By comparison, using approximations to prepare the single-qubit catalyst $\ket{v}$, we get a $T$-count of
\[
6\log_2(1/\varepsilon) + 4m \approx 300 + 4m
\] 
with the given embedding. In the limit of large $m$, the ratio of these two costs has an asymptotic scaling of
\[
    \frac{4}{3\cdot(50+\log_2 m)}.
\]
Ignoring the $\log_2 m$ term, this reduces the T-count by 97\%
compared to the standard approach. Including this term for a
reasonable number of $E$ gates (say $m=2^{20}\approx 10^6$), this
reduces the T-count by 98\% over the standard method. Were we to
consider the arbitrary precision limit, this value becomes arbitrarily
close to 100\%. This highlights the power of catalytic embeddings to
reduce gate counts in practice.

\subsection{The Quantum Fourier Transform}
\label{sect:QFT}

The quantum Fourier transform (QFT) on $n$ qubits is the unitary
operation given by the matrix
$\frac{1}{2^{n/2}}[\omega^{jk}]_{j,k=0}^{n-1}$ where $\omega = e^{2\pi
  i/2^n}$. It is well-known that the QFT can be realized as the
circuit below, where $R_k$ is the $2\times 2$ diagonal matrix
$R_k=\diag(1,e^{2\pi i/2^k})$ \cite{NC}.
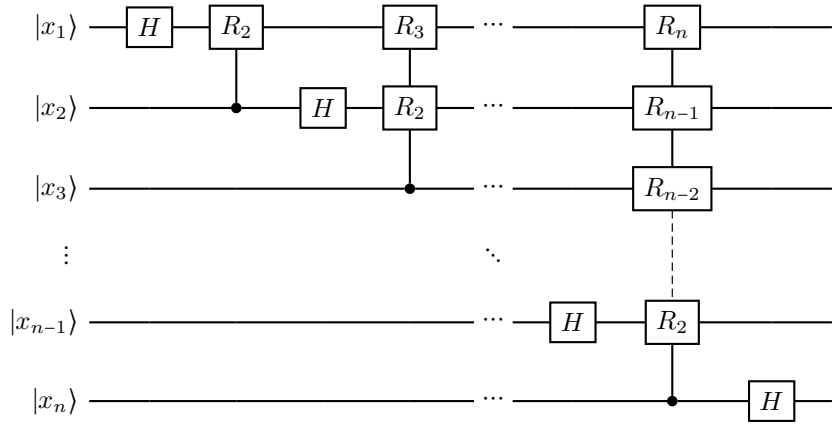
    \begin{center}
        \begin{quantikz}
            \lstick{$\ket{x_1}$} & \gate{H} & \gate{R_2} & \qw & \gate{R_3} & \push{~\cdots~} & \qw & \gate{R_{n}} & \qw & \qw\rstick{}\\
            \lstick{$\ket{x_2}$} & \qw & \ctrl{-1} & \gate{H} & \gate{R_2} & \push{~\cdots~} & \qw & \gate{R_{n-1}} & \qw  & \qw\rstick{}\\
            \lstick{$\ket{x_3}$} &\qw & \qw & \qw & \ctrl{-2} & \push{~\cdots~} & \qw & \gate{R_{n-2}}\arrow[dd,dashed,dash]\vqw{-2} & \qw & \qw\rstick{}\\
            \lstick{$\myvdots$} & & &  & & \myddots &&& \rstick{}\\
            \lstick{$\ket{x_{n-1}}$} & \qw & \qw & \qw & \qw & \push{~\cdots~} & \gate{H} & \gate{R_2} & \qw & \qw\rstick{} \\
            \lstick{$\ket{x_n}$} & \qw & \qw & \qw & \qw & \push{~\cdots~} & \qw & \ctrl{-1} & \gate{H} & \qw\rstick{} 
        \end{quantikz}
    \end{center}

We construct a standard catalytic embedding to implement this circuit
using the gate set $\langle H,X,CX,CCX\rangle$.  We then show how to
use $n$ additional $X$ gates and the same embedded circuit to
implement the inverse quantum Fourier transform.

Our goal is to reduce the cost of the expensive $Z$-rotations. To do
so, we will find an embedding $\phi:\unitary(\Z[e^{2\pi
    i/2^n}])\rightarrow\unitary(\N)$. We first construct a sequence of
standard catalytic embeddings
\[
\unitary(\ringR_n)\rightarrow
\unitary(\ringR_{n-1})\rightarrow\dots
\rightarrow \unitary(\ringR_k)\rightarrow
\unitary(\ringR_{k-1})\rightarrow
\dots\rightarrow\unitary(\ringR_2)\rightarrow\unitary(\ringR_1)\]
 where $\ringR_k = \Z[e^{2\pi i/2^k}]$. Let us begin with $\phi_k:\unitary(\ringR_k)\rightarrow
\unitary(\ringR_{k-1})$. The minimal polynomial for $e^{2\pi i/2^k}$ over $\ringR_{k-1}$ is $p_k(x)=x^2 - e^{2\pi i/2^{k-1}}$. Define
\[
\Lambda_k = \begin{bmatrix}
    0 & 1\\
    e^{2\pi i/2^{k-1}} & 0
\end{bmatrix} \qquad \mbox{and} \qquad \ket{\psi_k} = \frac{1}{\sqrt{2}}\begin{bmatrix}
    1\\
    e^{2\pi i/2^{k}}
\end{bmatrix}.
\]
The matrix $\Lambda_k$ is a normal companion matrix for $p_k$ and has
$\ket{\psi_k}$ as its corresponding catalyst. We have the following
embedding of $R_k$:\[
\begin{bmatrix}
    1 & 0\\
    0 & e^{2\pi i/2^k}
\end{bmatrix}\longmapsto\begin{bmatrix}
    1 & 0 & 0 & 0\\
    0 & 1 & 0 & 0\\
    0 & 0 & 0 & 1\\
    0 & 0 & e^{2\pi i /2^{k-1}} & 0
\end{bmatrix}
\]
Using controlled gates, we can represent this embedding
diagrammatically as below.
    
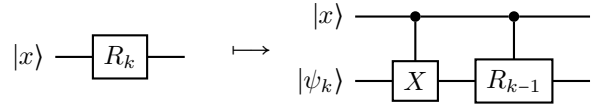
\begin{center}
        \begin{quantikz}
            \lstick{$\ket{x}$} & \gate{R_k} & \qw & 
        \end{quantikz}
        $\longmapsto$
        \begin{quantikz}
             \lstick{$\ket{x}$} & \ctrl{1} & \ctrl{1} &  \qw\\
             \lstick{$\ket{\psi_k}$} & \gate{X} & \gate{R_{k-1}} & \qw
        \end{quantikz}
    \end{center}
Following this sequence of embeddings to its end, we obtain the
following following embedding for $R_k$:
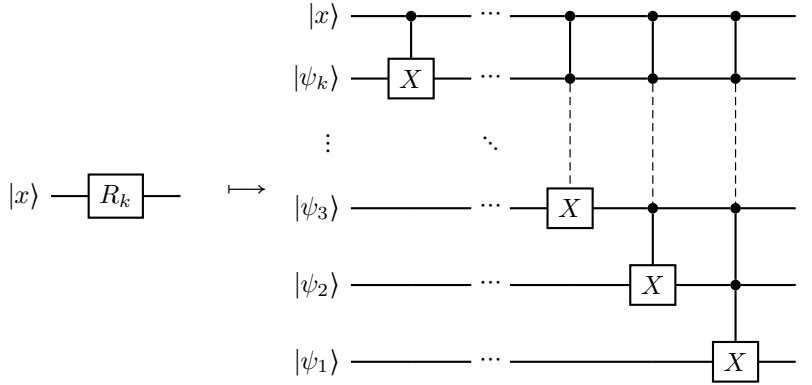
\begin{center}
        \begin{quantikz}
            \lstick{$\ket{x}$} & \gate{R_k} & \qw & 
        \end{quantikz}
        $\longmapsto$
    \begin{quantikz}
         \lstick{$\ket{x}$} & \ctrl{1} & \push{~\cdots~} & \ctrl{1} & \ctrl{1}  & \ctrl{1} &\qw\rstick{}\\
         \lstick{$\ket{\psi_k}$} & \gate{X} & \push{~\cdots~} & \phase{}\arrow[dd,dashed,dash] & \phase{}\arrow[dd,dashed,dash] & \phase{}\arrow[dd,dashed,dash] &\qw\rstick{}\\
         \lstick{$\myvdots$} & & \myddots & & & & \\
         \lstick{$\ket{\psi_3}$} & \qw & \push{~\cdots~} & \gate{X} & \ctrl{1} & \ctrl{1}   &\qw\rstick{}\\
         \lstick{$\ket{\psi_2}$} &\qw & \push{~\cdots~} & \qw & \gate{X} & \ctrl{1}  &\qw\rstick{}\\
         \lstick{$\ket{\psi_1}$} &\qw & \push{~\cdots~} & \qw & \qw & \gate{X} &\qw\rstick{}
    \end{quantikz}.
\end{center}

In our presentation above, we have gone one step further and embedded
$R_1$ (the $Z$ gate) as $CX$, which we note is not a linear catalytic
embedding (because $\N$ is not a number ring) but nonetheless
constitutes a catalytic embedding. On computational basis states, the
action of the circuit above can be described as a controlled
decrementer, where $\ket{x}$ is the control and the register
$\ket{\psi_1}\otimes \cdots\otimes \ket{\psi_k}$ holds the target
integer (represented as a bitstring for computational basis
states). We therefore simplify our notation and write the embedding as
below.
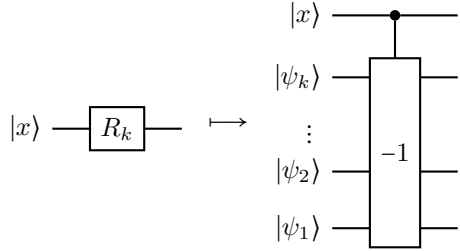
\begin{center}
\begin{quantikz}
    \lstick{$\ket{x}$} & \gate{R_k} & \qw
\end{quantikz} ~~$\longmapsto$
    \begin{quantikz}
         \lstick{$\ket{x}$}& \ctrl{1} & \qw\rstick{}\\
        \lstick{$\ket{\psi_k}$} & \gate[wires=4,nwires=2]{-1} & \qw\rstick{}\\
        \lstick{\myvdots}  & \rstick{}\\
        \lstick{$\ket{\psi_2}$}&  & \qw\rstick{}\\
        \lstick{$\ket{\psi_1}$}&  & \qw\rstick{}
    \end{quantikz}
\end{center}
By \cref{prop:directsum}, linear catalytic embeddings are well-behaved
with respect to direct sums, and so the embedding of a controlled
operation is a controlled operation. Similarly, despite not being a
linear catalytic embedding the final catalytic embedding also respects
direct sum structure. Putting all this together, we get the following
embedding of the quantum Fourier transform:
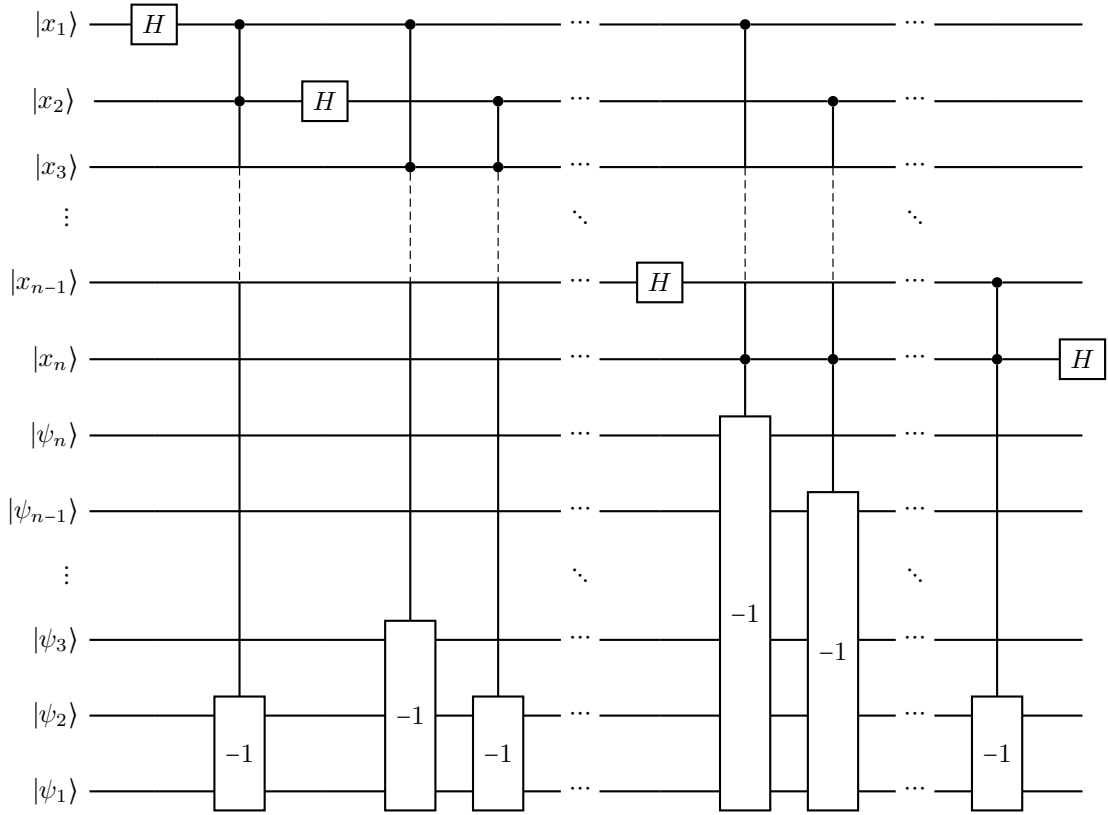
\begin{center}
    \begin{quantikz}
        \lstick{$\ket{x_1}$} & \gate{H} & \ctrl{1} & \qw & \ctrl{2} & \qw & \push{~\cdots~} & \qw & \ctrl{2} & \qw & \push{~\cdots~} & \qw & \qw\rstick{}\\\
        \lstick{$\ket{x_2}$} & \qw & \ctrl{1} & \gate{H} & \qw & \ctrl{1} & \push{~\cdots~} & \qw & \qw & \ctrl{1} & \push{~\cdots~} & \qw & \qw\rstick{}\\
        \lstick{$\ket{x_3}$}& \qw & \qw \arrow[dd,dashed,dash] & \qw & \phase{}\arrow[dd,dashed,dash] & \phase{}\arrow[dd,dashed,dash] & \push{~\cdots~} & \qw & \qw\arrow[dd,dashed,dash] & \qw\arrow[dd,dashed,dash] & \push{~\cdots~} & \qw & \qw\rstick{}\\
        \lstick{$\myvdots$}&&&&&&\myddots&&&&\myddots\\
        \lstick{$\ket{x_{n-1}}$} & \qw & \qw \vqw{6} & \qw & \qw \vqw{5} & \qw \vqw{6} & \push{~\cdots~} & \gate{H} & \qw \vqw{1} & \qw \vqw{1} & \push{~\cdots~} & \ctrl{1} & \qw\rstick{}\\
        \lstick{$\ket{x_n}$}&\qw&\qw&\qw&\qw&\qw&\push{~\cdots~} & \qw & \ctrl{1} & \ctrl{2} & \push{~\cdots~} & \ctrl{5} & \gate{H}\rstick{}\\
        \lstick{$\ket{\psi_n}$} & \qw & \qw & \qw & \qw & \qw & \push{~\cdots~} & \qw & \gate[wires=6,nwires=3]{-1} & \qw& \push{~\cdots~} & \qw & \qw\rstick{}\\
        \lstick{$\ket{\psi_{n-1}}$} & \qw & \qw & \qw & \qw & \qw & \push{~\cdots~} & \qw & & \gate[wires=5,nwires=2]{-1} & \push{~\cdots~} & \qw & \qw\rstick{}\\
        \lstick{$\myvdots$}&&&&&&\myddots&&&&\myddots\\
        \lstick{$\ket{\psi_3}$}&\qw&\qw&\qw&\gate[wires=3]{-1}&\qw&\push{~\cdots~}& \qw && & \push{~\cdots~} & \qw & \qw\rstick{}\\
        \lstick{$\ket{\psi_2}$} & \qw & \gate[wires=2]{-1} &\qw& &\gate[wires=2]{-1} &\push{~\cdots~}& \qw && & \push{~\cdots~} & \gate[wires=2]{-1} & \qw\rstick{}\\
        \lstick{$\ket{\psi_1}$} & \qw &  & \qw & & & \push{~\cdots~} & \qw & & & \push{~\cdots~} & & \qw\rstick{}
    \end{quantikz}
\end{center}
Each of these controlled-decrement circuits can be implemented with
$X$, $CX$, and $CCX$ gates \cite{shende03}, and so we have implemented
the quantum Fourier transform using the gate set $\langle
H,X,CX,CCX\rangle$ in the presence of catalysts.

Compiled as a circuit over Clifford+$T$ instead, 
the above implementation of the quantum Fourier transform is
equivalent to the $T$-count efficient circuit given in
\cite{gidney,nam2020approximate}. Each individual decrement can be
seen as an inverse adder controlled by the bottom control qubit and
subtracting the top control qubit from the ancilla qubits. Each
sequence of decrements can be seen as an inverse adder taking as input
the binary number represented by the top control bits and subtracting
it from the ancilla bits. As an $n$-bit adder can be implemented with
linear $T$ complexity, and noting that $\ket{\psi_k}=R_k H\ket{0}$,
the entire circuit may be implemented over Clifford+$T$ with $T$ count
\[
    O(n^2 + n\log_2(1/\epsilon))
\]
compared to the standard approach of approximation, which would
require $T$-count
\[
    O(n^2\log_2(1/\epsilon)).
\]
While this implementation of the QFT has been previously derived using
phase gradients \cite{gidney,nam2020approximate}, we have shown that
the more general framework of catalytic embeddings suffices to
reproduce it.

While it might appear that catalytic embeddings have merely reproduced
the best known constructions of the QFT, we can in fact glean
additional insight into the structure of those constructions. In light
of \cref{cor:RSinterplay}, we know that the catalyst
$\ket{\psi_1}\otimes\cdots\otimes\ket{\psi_n}$ produced in this
construction should have orthogonal counterparts corresponding to
alternative embeddings of $e^{2\pi i/2^n}$ in $\C$. For example, by
applying an $X$ gate to each $\ket{\psi_k}$ we see
that \[X\frac{1}{\sqrt{2}}\begin{bmatrix} 1\\ e^{2\pi/2^k}
\end{bmatrix} = \frac{1}{\sqrt{2}}\begin{bmatrix}
    e^{2\pi/2^k}\\
    1
\end{bmatrix} \sim \frac{1}{\sqrt{2}}\begin{bmatrix}
    1\\
    e^{-2\pi/2^k}
\end{bmatrix}.
\]
The resulting state is orthogonal to the original state since
$\bra{\psi_2} X \ket{\psi_2} = 0$, and by inspection such a tensor
product of $X$ gates induces the complex conjugation automorphism on
$\omega$. Therefore, $X\ket{\psi_1}\otimes\cdots\otimes X\ket{\psi_n}$
is one such alternative catalyst, and it is precisely such that it
maps each $\omega$ in the QFT unitary to $\omega^\dagger$, inducing
the complex conjugation automorphism on the circuit which happens to
be equivalent to the inverse QFT. Thus, using the \emph{same} embedded
circuit along with $n$ $X$ gates, we can implement the inverse QFT.

\section{Perspectives}
\label{sec:conc}

In this paper, we laid the foundations for the theory of catalytic
embeddings. We believe that catalytic embeddings may find a variety of
applications in the study of quantum circuits and, more generally, in
the theory and practice of fault-tolerant quantum computation. As
discussed in \cref{sec:standard}, there are cases where catalytic
embeddings reproduce or beat existing quantum circuit constructions
for specific operations. We are eager to see what other algorithmic
primitives can be improved with catalytic embeddings. Approximate and
exact synthesis methods also seem like prime candidates to bolster
with the power of catalytic embeddings. To make the most of this
framework, it is important to provide constructive methods for
producing catalytic embeddings. In follow-up work
\cite{ringsandfields}, we provide such constructive methods in many
cases of interest.

The structure-preserving nature of catalytic embeddings may provide
insights into a number of open questions. Firstly, one may be able to
use catalytic embeddings to better understand gate sets. By embedding
a poorly understood gate set into a well-understood one (such as the
Toffoli-Hadamard gate set), one could in principle transform results
about the latter into results about the former. This approach may help
in characterizing gate sets, finding relations for circuits, and
deriving asymptotic lower bounds for resources. Further afield,
catalytic embeddings seem to be a natural tool with which to tackle
long-standing open questions about the Clifford hierarchy. Indeed,
catalytic embeddings were (in part) born out of generalizing gate
teleportation protocols. Even farther afield, there seems to be a
growing body of evidence that various approaches to achieving
fault-tolerant quantum computation share important
properties. Catalytic embeddings may help in understanding what
unifies these different approaches.

\section{Acknowledgements}
\label{sec:ack}

The circuit diagrams in this paper were typeset using Quantikz
\cite{quantikz}. This work was funded in part by Naval Innovative
Science and Engineering funding.

\bibliography{embeddings}

\end{document}